\documentclass{article}

\usepackage[utf8]{inputenc} 
\usepackage[T1]{fontenc}    
\usepackage{amsmath,amsfonts,amssymb}
\usepackage{nicefrac}       
\usepackage{microtype}      
\usepackage{bm}
\usepackage{multirow}
\usepackage{hyperref}
\usepackage{url}
\usepackage{graphicx}
\usepackage{geometry,bookmark}

\usepackage{epstopdf}
\usepackage{algorithm}
\usepackage{algorithmicx}
\usepackage{algpseudocode}
\usepackage{amsthm}
\usepackage[dvipsnames]{xcolor}

\usepackage[numbers,square,comma,sort&compress]{natbib}

\usepackage[capitalise]{cleveref}
\usepackage{comment}

\usepackage{booktabs} 
\usepackage{soul}  
\usepackage[shortlabels]{enumitem}
\graphicspath{{}{figures/}}
\allowdisplaybreaks

\definecolor{OliveGreen}{rgb}{0,0.6,0}

\def\lV{\left\lVert }
\def\rV{\right\lVert }
\def\lv{\lvert }

\def\l{\left\langle}
\def\r{\right\rangle}

\def\H{\mathcal{H} }
\def\D{\mathcal{D} }
\def\G{\mathcal{G} }
\def\I{\mathcal{I} }

\def\L{{\mathcal{L}} }

\def\R{{\mathcal{R}} }

\def\B{\mathcal{B} }
\def\P{\mathcal{P} }
\def\E{\mathcal{E} }
\def\Z{\mathcal{Z} }

\def\Y{\mathcal{Y} }

\def\Y{\mathcal{Y} }
\def\C{\mathbb{C} }

\def\bz{\bm{z} }
\def\bx{\bm{x} }
\def\by{\bm{y} }
\def\bw{\bm{w} }
\def\bs{\bm{s} }
\def\bf{\bm{f} }
\def\bv{\bm{v} }
\def\cO{\mathcal{O} }
\def\BL{{\bm{L}} }
\def\BR{{\bm{R}} }
\def\BU{{\bm{U}} }
\def\BV{{\bm{V}} }
\def\BSigma{\bm{\Sigma} }
\def\BQ{{\bm{Q}} }

\def\BX{{\bm{X}} }
\def\BDeltaR{\bm{\Delta_R}}
\def\BDeltaL{\bm{\Delta_L}}
\def\SVD{\mathrm{SVD} }
\def\fro{\mathrm{F}}
\def\rank{\mathrm{rank}}

\DeclareMathOperator*{\argmin}{\mathrm{arg\,min}}
\DeclareMathOperator*{\minimize}{\mathrm{minimize}}
\DeclareMathOperator*{\subjectto}{\mathrm{subject~to}}

\newtheorem{theorem}{Theorem}
\newtheorem{lemma}{Lemma}
\newtheorem{definition}{Definition}

\newtheorem{assumption}{Assumption}
\newtheorem{corollary}{Corollary}
\newtheorem{remark}{Remark}
\Crefname{assumption}{Assumption}{Assumptions}
\Crefname{example}{Example}{Examples}
\Crefname{remark}{Remark}{Remarks}

\begin{document}

\title{Structured Gradient Descent for Fast Robust Low-Rank Hankel Matrix Completion}

\author{
HanQin Cai\thanks{Department of Mathematics, University of California, Los Angeles, Los Angeles, CA 90095, USA. (email: hqcai@math.ucla.edu).}
\and Jian-Feng Cai\thanks{Department of Mathematics, The Hong Kong University of Science and Technology, Clear Water Bay, Kowloon, Hong Kong SAR, China. (email:jfcai@ust.hk, jyouab@connect.ust.hk) .}
\and Juntao You\footnotemark[2] $^,$\thanks{Corresponding author.}
}


\date{}
\maketitle

\begin{abstract}
We study the robust matrix completion problem for the low-rank Hankel matrix, which detects the sparse corruptions caused by extreme outliers while we try to recover the original Hankel matrix from partial observation.
In this paper, we explore the convenient Hankel structure and propose a novel non-convex algorithm, coined Hankel Structured Gradient Descent (HSGD), for large-scale robust Hankel matrix completion problems.
HSGD is highly computing- and sample-efficient compared to the state-of-the-arts. The recovery guarantee with a linear convergence rate has been established for HSGD under some mild assumptions.
The empirical advantages of HSGD are verified on both synthetic datasets and real-world nuclear magnetic resonance signals.
\end{abstract}

%

\section{Introduction} \label{sec:introduction}
Recently, the problems of Hankel matrix have received much attention. A complex-valued Hankel matrix has identical values on every antidiagonal:
\begin{align} \label{eq:Hankel definition}
\H\bx =
\begin{bmatrix}
x_1 & x_2  & \cdots & x_{n_2}\\
x_2 & x_3  & \cdots & x_{n_2+1}  \\
\vdots &  \vdots & \cdots  &\vdots\\
x_{n_1} & x_{n_1+1}  & \cdots   &x_{n}
\end{bmatrix}\in\C^{n_1\times n_2},
\end{align}
where $\bx=[x_1,x_2,\cdots,x_n]^\top\in\C^n$ is a  vector consisting of the distinct entries of the Hankel matrix, and the linear operator $\H:\C^n\rightarrow\C^{n_1\times n_2}$ is called Hankel mapping. Note that we always have $n=n_1+n_2-1$ regardless of the shape of the Hankel matrix. Hence, a Hankel matrix can always be efficiently mapped from a much smaller vector.

The low-rank Hankel matrix has arisen in various applications; for instance, nuclear magnetic resonance (NMR) spectroscopy \cite{holland2011fast,nguyen2013denoising,qu2015accelerated}, medical imaging \cite{jin2016general,jacob2020structured}, seismic imaging \cite{oropeza2011simultaneous,chen2021exact}, and autoregression \cite{lee2020non,lobos2020autoregression}. In real-world applications, there are often two major challenges for data analysis:
\begin{enumerate}
    \item \textit{Missing data.} Due to hardware limitations or time constraints, only partial data can be obtained \cite{tropp2009beyond}. By the structure of the Hankel matrix, we naturally know all the values on an antidiagonal once we observe one entry on it. In this sense, the partial observation of a Hankel matrix must also be Hankel structured.
    \item \textit{Extreme outliers.} Due to hardware malfunction, the recorded data are often sparsely corrupted by extreme outliers \cite{xi2008baseline}. Note that the outliers are usually Hankel structured as well due to the nature of these applications. In fact, the outliers who corrupt only a few entries on an antidiagonal can be easily detected by comparing all values on that antidiagonal. Thus, we are only interested in the Hankel structured outliers that destroy some antidiagonals completely.
\end{enumerate}
In the NMR spectroscopy example, observing the full signal is very time-consuming and can take up to months. Thus, the researchers often deal with the partially observed signal due to time constraints. On the other hand, impulse noise often causes sparse corruptions on the observed signal, which can be viewed as additive outliers.

In this paper, we aim to handle both these challenges simultaneously by solving the \textbf{R}obust Low-rank \textbf{H}ankel Matrix \textbf{C}ompletion (RHC) problem:
Given the sparsely corrupted partial observations on the Hankel matrix, we want to recover the original low-rank Hankel matrix. Note that the observations are Hankel structured since observing one entry gives the information for an entire antidiagonal. In addition, we are only interested in the sparse corruptions with Hankel structured support since (i)  unstructured outliers can be easily removed when we sample multiple entries on the same antidiagonal,\footnote{Applying a simple median filter will remove the outlier(s) as long as less than half of the observations on this antidiagonal were corrupted.
}
or (ii) corrupting very few entries is equivalent to corrupting the entire antidiagonal if only very few entries were sampled. Moreover, in many applications, e.g., NMR spectroscopy, only one entry (viz.~one copy) may be observed for each antidiagonal. Thus, adding outliers to such observations is the same as adding Hankel structured outliers.

With all the participants in the Hankel structure, we can write the RHC problem in the following form: Consider
$$\H\by=\H\bm{x}^{\natural}+\H\bm{w}^{\natural},$$
where  $\H \bm{x}^{\natural} \in \C ^{n_1 \times n_2}$ is the underlying low-rank Hankel matrix and $\H \bm{w}^{\natural}\in \C^{n_1 \times n_2}$ is
the Hankel structured sparse corruption matrix. The problem is to recover $\H \bm{x}^{\natural}$ (or equivalently $\bm{x}^{\natural}\in \C^n$) from the sparsely corrupted partial observation
\begin{equation}  \label{problem1}
 \H\Pi_{\Omega}\by=\H\Pi_{\Omega}(\bx^\natural+\bw^\natural),
\end{equation}
where $\Omega\subseteq [n]$ is the index set of the observed entries in corresponding vector form.  $\Pi_{\Omega}:\C^n\rightarrow\C^n$ is the sampling operator defined by $[\Pi_{\Omega}\by]_i = y_i$ if $i\in\Omega$, or otherwise $[\Pi_{\Omega}\by]_i=0$.
Naturally, we consider the loss function:
\begin{align} \label{eq:matrix loss}
  \frac{1}{2p}\|\H\Pi_{\Omega}\by-\H\Pi_{\Omega}(\bx+\bw) \|_\fro^2,
\end{align}
where $p$ is the probability of an entry being observed. One can see that the RHC problem can be efficiently described in a more straightforward vector form  $\Pi_{\Omega}\by=\Pi_{\Omega}(\bx^\natural+\bw^\natural)$.
However, the corresponding vector loss function
$ \frac{1}{2}\|\Pi_{\Omega}\by-\Pi_{\Omega}(\bx+\bw) \|_2^2$
is not equivalent to the original Hankel matrix loss function \cref{eq:matrix loss}, because antidiagonals contain different numbers of repeated entries. Thus, we must adjust the weights of the entries in the vector loss function. To this end, we introduce the reweighting operator $\D:\C^n\rightarrow\C^n$ defined as $[\D\bx]_i=\sqrt{\varsigma_i}x_i$, where $\varsigma_i$ is the number of entries on the $i$-th antidiagonal of an $n_1\times n_2$ matrix, so that
\begin{align}
   \frac{1}{2p}\|\Pi_{\Omega}\D\by-\Pi_{\Omega}\D(\bx+\bw) \|_2^2
\end{align}
is equivalent to \cref{eq:matrix loss}. For the interested reader, we extend further discussion on the necessity of the reweighting operator in the supplementary material. For ease of presentation, we introduce the reweighted notations:
\begin{align}
    \bm{f}:=\D \by, \quad \bm{z}:=\D \bm{x}, \quad \bm{s}:=\D \bm{w}, \quad \G:=\H\D^{-1}.
\end{align}
Note that $\G\bf=\H\by$, $\G\bz=\H\bx$, and $\G\bs=\H\bw$. Combined with the low-rank constraints on $\bx$ and the sparse constraints on $\bw$, we have the vector formula that is equivalent to the original Hankel matrix formula:
\begin{equation} \label{eq:RHC1}
    \begin{split}
        \minimize_{\bz,\bs\in\C^n}\ & ~\frac{1}{2p}\|\Pi_{\Omega}\bf-\Pi_{\Omega}(\bz+\bs) \|_2^2 \cr
        \subjectto & ~ \rank\left(\G \bm{z}\right)=r,  \quad \Pi_{\Omega}\bm{s} \textnormal{ is }\alpha p\textnormal{-sparse},
    \end{split}
\end{equation}
where $\alpha p$-sparsity will be formally defined as \cref{amp:sparsity} later. Note that $\D ^2=\H^*\H$ and $\G^*\G=\I$ where $(\cdot)^*$ denotes the adjoint operator. Hence, RHC can be viewed as robust matrix completion on an orthonormal basis.

\subsection{Assumptions and notation}
In this subsection, we present the problem assumptions and notation that will be used for the rest of the paper. We start with the sampling model for the observations.

\begin{assumption}[Bernoulli sampling] \label{amp:Bernoulli}
The set of sampling index $\Omega$ is drawn by the Bernoulli model  with probability $p$. 
That is, the $i$-th entry of $\by$ is observed with probability $p$ independent of all others.  
\end{assumption}

Note that the entries on an antidiagonal are repeated. Thus, for sampling efficiency, we will sample no more than once on each antidiagonal for sampling efficiency. Hence, our sampling model is presented on the vector that consists of the distinct elements of the Hankel matrix. In practice, the actual probability $p$ is usually unknown, and it is common to take $p=|\Omega|/n$.
For the reader's interest, we highlight that a similar method and theorem can also be developed for the uniform sampling model with no extra difficulty. 

\begin{assumption}[$\mu$-incoherence] \label{amp:incoherence}
The rank-$r$ Hankel matrix $\H\bx^\natural=\G \bz^{\natural}\in\C^{n_1\times n_2}$ is $\mu$-incoherent, i.e., there exists a constant $\mu$ such that
\begin{equation*}
\| \BU^{\natural} \|_{2,\infty}\le \sqrt{\mu c_s r/n} \quad \textnormal{ and } \quad \| \BV^{\natural} \|_{2,\infty}\le \sqrt{\mu c_s r/n},
\end{equation*}
where $\BU^{\natural}\BSigma^{\natural}\BV^{\natural *}$ is the compact SVD of $\H\bx^\natural$ and $c_s=\max\{{n}/{n_1},{n}/{n_2}\}$.
\end{assumption}

The assumption of incoherence was first introduced in \cite{candes2009exact} and has been widely used in RPCA and matrix completion studies. The parameter $\mu$ describes how evenly the information is distributed among the entries of the matrix, and its value is small in a well-conditioned dataset. Empirically, many applications related to low-rank matrices, such as face recognition \cite{wright2008robust} and video background subtraction \cite{cai2019accaltproj}, satisfy the incoherence condition with small $\mu$.
In this paper, \cref{amp:incoherence} is a Hankel variation of the standard incoherence. In the applications of Hankel matrix, the parameter $\mu$ is often very small. For example, in the application of spectrally sparse signal, $\mu=\cO(1)$ if the signal is well separated in the frequency domain \cite[Theorem~2]{liao2016music}. 
Note that the constant $c_s$ describes how square the Hankel matrix is. Ideally, when $n$ is fixed, $n_1\approx n_2\approx n/2$ gives a relatively easier problem.

\begin{assumption}[$\alpha p$-sparsity] \label{amp:sparsity}
Under the Bernoulli model with probability $p$, i.e., \cref{amp:Bernoulli}, the vector of observed outliers is $\alpha p$-sparse. That is,
\begin{equation*}
    \|\Pi_\Omega \bs^\natural\|_0=\|\Pi_\Omega \bw^\natural\|_0\leq \alpha p n.
\end{equation*}
\end{assumption}

Following the same argument for \cref{amp:Bernoulli}, the observed outliers are also presented in the vector form. Note that \cref{amp:sparsity} holds with high probability provided $\bs^\natural$ is $(\alpha/2)$-sparse \cite[Proposition~4.5]{cai2021rcur}. This assumption also implies that $\H\Pi_\Omega \bs^\natural$ has no more than $\alpha p n$ non-zero entries in each row and column.

In the rest of this subsection, we describe the notation that we will use throughout the paper.
We use regular lowercase letters for scalars, bold lowercase letters for vectors, bold capital letters for matrices, and calligraphic letters for operators. For a vector $\bm{v}$, let $\|\bm{v}\|_0$, $\|\bm{v}\|_2$, and $\|\bm{v}\|_\infty$ denote $\ell_0$-, $\ell_2$-, and $\ell_\infty$-norms, respectively.
For a matrix $\bm{M}$, $\|\bm{M}\|_{2,\infty}$ denotes the largest $\ell_2$-norm of the rows, $\|\bm{M}\|_\infty$ denotes the largest magnitude in the entries, $\sigma_i(\bm{M})$ denotes the $i$-th singular value, $\|\bm{M}\|_2$ denotes its spectral norm, and $\|\bm{M}\|_\fro$ denotes its Frobenius norm. $\|\P\|$ denotes the operator norm of the linear map $\P$. For both vectors and matrices, $\overline{(\cdot)}$, $(\cdot)^\top$, $(\cdot)^*$, and $\langle\cdot,\cdot\rangle$ denote conjugate, transpose, conjugate transpose, and inner product, respectively. Moreover, $\sigma_i^\natural$ denotes the $i$-th singular value of the underlying Hankel matrix $\H\bx^\natural$ and $\kappa=\sigma_1^\natural/\sigma_r^\natural$ denotes its condition number. Later in \cref{sec:proofs}, we introduce some additional notations used only in the analysis.

\subsection{Related work and main contributions}
The decomposition problem for generic low-rank and sparse matrices is known as robust principal component analysis (full observation) or robust matrix completion (partial observation), which has been widely studied in both theoretical and empirical aspects \cite{candes2011robust,chandrasekaran2011rank,chen2013low,netrapalli2014non,yi2016fast,cherapanamjeri2017nearly,cai2019accaltproj,cai2021ircur,cai2021rcur,cai2021lrpca,cai2021rtcur,hamm2022riemannian}. However, without utilizing the convenient Hankel structure, the generic matrix approaches are sub-optimal on RHC problems, in terms of robustness, sample complexity, and computational efficiency.

Recently, dedicated methods have been proposed for low-rank Hankel matrix problems. Robust Enhanced Matrix Completion (Robust-EMaC) \cite{chen2014robust} relaxes the non-convex RHC problem with a convex formula where the nuclear norm and $\ell_1$-norm are used to enforce low-rank and sparsity constraints, respectively. While Robust-EMaC has the state-of-the-art sampling complexity, requiring merely $\cO(c_s^2\mu^2r^2\log n)$ samples, it does not provide an efficient numerical solver.\footnote{The original paper uses semidefinite programming to solve the convex model, which is as expensive as $\cO(n^6)$. The first-order method can improve this to $\cO(n^3)$ flops per iteration, which is still too expensive.} Note that Robust-EMaC tolerates a small constant portion of outliers if the support of outliers is randomly distributed, which is more restrictive than \cref{amp:sparsity}.
Other convex approaches \cite{tang2013compressed,bhaskar2013atomic,chen2021exact} have similar computational challenging when the problem scale is large.
Later, more provable non-convex methods with linear convergence are introduced. \cite{cai2018spectral,cai2019fast} propose two fast algorithms that both solve a Hankel matrix completion problem in $\cO(r^2n +rn\log n)$ flops per iteration. Unfortunately, they are not designed to handle extreme outliers.  Structured Alternating Projection (SAP) \cite{zhang2018correction} is an alternating projection-based algorithm that efficiently removes outliers from the Hankel matrix, but the theoretical guarantee is only established under full observation. The computational complexity of SAP is $\cO(r^2n\log n)$ per iteration, which is later improved to $\cO(r^2n +rn\log n)$ by its accelerated version, namely ASAP \cite{cai2021asap}; however, ASAP only focuses on the fully observed cases. In the follow-up work \cite{zhang2019correction}, SAP is modified for exploring the setting of partial observation, which we call PartialSAP. PartialSAP has the same computational complexity as SAP, requiring $\cO(c_s^2\mu^2r^3\log^2 n \log \varepsilon^{-1})$ samples due to its iterative re-sampling requirement in theory. PartialSAP can tolerate $\alpha\lesssim \cO(1/(c_s \mu r))$ fraction outliers\footnote{$a \lesssim \cO(b)$ means there exists an absolute constant $C > 0$ such that $a\leq C\cdot b$.}, under the same \cref{amp:sparsity}.

In this work, we propose a novel non-convex algorithm, coined Hankel Structured Gradient Descent (HSGD), for solving large-scale RHC problems. Our main contributions are:
\begin{enumerate}
    \item HSGD is computing-efficient. Its computational complexity is $\cO(r^2 n+rn\log n)$ flops per iteration---better than the current state-of-the-art PartialSAP.
    \item HSGD is sample-efficient. Without the requirement of iterative re-sampling, its overall sample complexity is $\cO(\max\{c_s^2\mu^2r^2 \log n, c_s\mu \kappa^3r^2 \log n\})$---tied with the state-of-the-art Robust-EMaC if the problem is well-conditioned.
    \item The recovery guarantee has been established. Under some mild assumptions, we show HSGD converges linearly to the ground truth. In particular, the theoretical outlier tolerance is $\alpha\lesssim \cO(1/\max\{c_s \mu\kappa^{3/2}r^{3/2}, c_s \mu r\kappa^2\} )$.
    \item The empirical advantages of HSGD are verified by both synthetic datasets and  real-world NMR signals. We observe that HSGD outperforms Robust-EMaC and PartialSAP for both speed and recoverability.
\end{enumerate}

\section{Proposed method}
In this section, we propose a highly efficient non-convex algorithm for the RHC problem \eqref{eq:RHC1}, coined Hankel Structured Gradient Descent (HSGD). The first major challenge in algorithm design is how to enforce the low-rank constraint on $\G(\bz)$. We reformulate the objective function so that the low-rank constraint can be avoided. Following \cite{yi2016fast}, we rewrite the rank-$r$ matrix on factorized space: $\G(\bz)=\BL\BR^*$ where $\BL\in\C^{n_1\times r}$ and $\BR\in\C^{n_2\times r}$. Thus, the low-rank constraint is automatically coded by the shapes of $\BL$ and $\BR$. Moreover, we add a penalty term $\frac{1}{2}\lV\left(\I-\G\G^*\right)\left(\BL \BR^*\right)\rV_\fro^2$ to make sure $\BL\BR^*$ is a Hankel matrix, because $\BL\BR^*$ is a Hankel matrix if and only if
$
    \left(\I-\G\G^*\right)\left(\bm{L} \bm{R}^*\right)=\bm{0}.
$
By replacing $\bz=\G^*(\BL\BR^*)$, we have
\begin{equation} \label{eq:def-F}
    \psi:=\psi\left(\bm{L},\bm{R};\bm{s}\right):=\frac{1}{2p}\lV \Pi_{\Omega}\left(\G^*\left(\bm{L} \bm{R}^*\right)+\bm{s}-\bm{f}\right)\rV_2^2 +\frac{1}{2}\lV\left(\I-\G\G^*\right)\left(\bm{L} \bm{R}^*\right)\rV_\fro^2.
\end{equation}
We also add another balance regularization
\begin{equation}
    \phi:=\phi\left(\BL,\BR\right):=\frac{1}{4}\lV \BL^*\BL-\BR^*\BR  \rV_\fro^2
\end{equation}
to encourage that $\BL$ and $\BR$ have the same scale, which is a common technique for factorized gradient descent under the asymmetric setting. \cite{ma2021beyond} suggests that the balance regularization $\phi$ may be removed in the matrix sensing problem if we have a very good initialization. However, we decide to keep this term since the initialization is usually more challenging in RHC problems. Putting all the pieces together, we have the non-convex loss function:
\begin{equation}\label{eq:def-L}
\ell:=\ell\left(\bm{L},\bm{R};\bm{s}\right):=\psi\left(\bm{L},\bm{R};\bm{s}\right)+\lambda \phi\left(\bm{L},\bm{R}\right),
\end{equation}
where $\lambda> 0$ is a regularization parameter.

Based on the reformulated loss function \eqref{eq:def-L}, we detail the proposed HSGD method in three steps: (a) initialization, (b) iterative updates on outliers, and (c) iterative updates on the low-rank Hankel matrix. 
For ease of presentation, we start the discussion with iterative updates, followed by initialization.

\vspace{0.05in}
\textbf{Iterative updates on outliers.}
With the sparsity assumption on the outlier vector $\bs$, we design a \textit{sparsification operator} to filter large-magnitude entries:
\begin{equation*}
    [\Gamma_{\tilde\alpha}(\bv)]_i=
    \begin{cases}
    v_i & \quad \textnormal{$|v_i|$ is one of the largest $\tilde\alpha n$ element in $|\bv|$} \\
    0 & \quad \textnormal{otherwise}\\
    \end{cases}
\end{equation*}
for any vector $\bv\in\C^{n}$ and $\tilde\alpha\in[0,1)$.
At the $(k+1)$-th iteration, we first compute the residues over the observed entries $\Pi_\Omega(\bf-\bz^{(k)})$ where $\bz^{(k)}=\G^*(\BL^{(k)}\BR^{(k)*})$ is the reweighted vector representing the estimated low-rank Hankel matrix from the previous iteration. Then, we keep the largest $\gamma_k\alpha p$-fraction of the residue as the outliers, i.e.,
\begin{equation} \label{eq:update_s}
\bs^{(k+1)}=\Gamma_{\gamma_k\alpha p}\big(\Pi_{\Omega}\big(\bf-\bz^{(k)}\big)\big),
\end{equation}
where $\gamma_k>1$ is a parameter that allows us to overestimate the outlier density a bit.

Recall that $\G^*$ is the adjoint operator of $\G=\H\D^{-1}$. For the $t$-th entry of $\G^*(\BL\BR^*)$,
\begin{equation*}
    [\G^*(\BL\BR^*)]_t= \sum_{j=1}^r [\G^*(\BL_{:,j}\BR_{:,j}^*)]_t = \sum_{j=1}^r \frac{1}{\sqrt{\varsigma_t}}\sum_{i_1+i_2=t+1} L_{i_1,j} \overline{R}_{i_2,j},
\end{equation*}
where $\varsigma_t$ is the number of entries on the $t$-th antidiagonal of the $n_1\times n_2$ Hankel matrix. This suggests that $\G^*(\BL\BR^*)$ can be computed via $r$ fast convolutions (i.e., via FFT). Thus, \eqref{eq:update_s} costs merely $\cO(rn \log n)$ flops.

\textbf{Iterative updates on low-rank Hankel matrix.}
After removing the estimated outliers, we simultaneously update the factors $\BL$ and $\BR$ of the low-rank Hankel matrix via one step of \textit{structured gradient descent}. The gradients with respect to $\BL$ and $\BR$ are calculated based on the loss function \eqref{eq:def-L}. Moreover, according to \cref{amp:incoherence}, we project the updates of $\BL$ and $\BR$ onto the convex sets
\begin{equation}\label{eq:projectionsets}
  \begin{split}
\L =\Big\{\bm{L}\,{\Big|}\,\lV \bm{L}\rV_{2,\infty}^2\le {\frac{2\mu r c_s}{n}}\big\|\tilde{\bm{L}}^{(0)}\big\|_2^2\Big\} \textnormal{ and }
\R =\Big\{\bm{R}\,{\Big|}\,\lV \bm{R}\rV_{2,\infty}^2\le {\frac{2\mu r c_s}{n}}\big\|\tilde{\bm{R}}^{(0)}\big\|_2^2\Big\}
  \end{split}
\end{equation}
respectively. The projection step ensures that the estimated low-rank matrix is always incoherent. Note that the ideal $\L$ and $\R$ should be defined with $\|\G(\bz)\|_2^{1/2}$. However, this information is usually unavailable to the user, so instead we use the initial estimations of $\BL$ and $\BR$. In summary, we have the projected gradient descent for updating the low-rank Hankel matrix:
\begin{equation} \label{eq:update_Hankel}
\begin{split}
\BL^{(k+1)}&=\Pi_\L\big(\BL^{(k)}-\eta\nabla_{\BL} \ell\big(\BL^{(k)},\BR^{(k)};\bs^{(k+1)}\big)\big),\\
\BR^{(k+1)}&=\Pi_\R\big(\BR^{(k)}-\eta\nabla_{\BR} \ell\big(\BL^{(k)},\BR^{(k)};\bs^{(k+1)}\big)\big),
\end{split}
\end{equation}
where $\eta>0$ is the step size. The complexity of \eqref{eq:update_Hankel} is dominated by computing the gradients. Notice that
\begin{align*}
    \nabla_\BL \ell = \G\big(p^{-1}\Pi_\Omega\left(\G^*\left(\BL\BR^*\right)+\bs-\bf\right)-\G^*(\BL\BR^*)\big)\BR  + \BL\left(\lambda\BL^*\BL+(1-\lambda)\BR^*\BR)\right).
\end{align*}
As discussed, computing $\G^*(\BL\BR^*)$ costs $\cO(rn \log n)$ flops, so does computing the vector $\bm{a}:=p^{-1}\Pi_\Omega\left(\G^*\left(\BL\BR^*\right)+\bs-\bf\right)-\G^*(\BL\BR^*)$. Next, $\G(\bm{a})\BR$ can be computed via another $r$ fast convolutions since a Hankel matrix can be viewed as a convolution operator.
Thus, the first term in $\nabla_\BL \ell(\BL,\BR;\bs)$ costs total $\cO(rn\log n)$ flops.
While the second term costs $\cO(r^2 n)$ flops, computing $\nabla_\BL \ell(\BL,\BR;\bs)$ costs merely $\cO(rn\log n+r^2 n)$ flops. The same argument applies to $\nabla_\BR \ell(\BL,\BR;\bs)$. Therefore, the update of the low-rank Hankel matrix is computationally efficient, in the complexity of $\cO(rn\log n+r^2 n)$.

\vspace{0.05in}
\textbf{Initialization.}
For a good initial guess, we modify the widely used spectral method \cite[Section~VIII]{chi2019nonconvex}. The first step is to detect the obvious (i.e., large) outliers from the observations:
\begin{equation*}
    \bs^{(0)}= \D\Gamma_{\alpha p}\left(\Pi_\Omega\by\right)= \D\Gamma_{\alpha p}\left(\D^{-1}\Pi_\Omega\bf\right),
\end{equation*}
where $\alpha p$ comes from \cref{amp:sparsity}. Next, we initial
\begin{align*}
\tilde{\bm{L}}^{(0)}=\BU^{(0)}\big(\BSigma^{(0)}\big)^{1/2} \qquad\textnormal{and}\qquad \tilde{\bm{R}}^{(0)}=\BV^{(0)}\big(\BSigma^{(0)}\big)^{1/2},
\end{align*}
where $\BU^{(0)}\BSigma^{(0)}\BV^{(0)*}$ is the rank-$r$ truncated SVD of $p^{-1}\G(\Pi_{\Omega}\bf-\bs^{(0)})$. Then, the convex sets $\L$ and $\R$ can be defined based the spectral norms of $\tilde{\bm{L}}^{(0)}$ and $\tilde{\bm{R}}^{(0)}$ (see \eqref{eq:projectionsets}). We immediately project $\tilde{\bm{L}}^{(0)}$ and $\tilde{\bm{R}}^{(0)}$ onto $\L$ and $\R$ to ensure their incoherence:
\begin{equation*}
\BL^{(0)}=\Pi_{\L}\tilde{\bm{L}}^{(0)} \qquad \textnormal{and} \qquad \BR^{(0)}=\Pi_{\R}\tilde{\bm{R}}^{(0)}.
\end{equation*}
This finishes the initialization.

The complexity of the initialization remains $\cO(rn\log n)$, which is dominated by the step of truncated SVD. Although a typical truncated SVD  costs $\cO(n^2 r)$ flops, it costs only $\cO(rn\log n)$ flops on a Hankel matrix since the involved Hankel matrix-vector multiplications can be computed via fast convolutions.

We summarize the proposed HSGD as \cref{algo:HSGD}. Although we are solving a matrix problem, HSGD never has to form the whole matrix due to the convenient Hankel structure---we only need to track the $n$ distinct entries in the reweighted vector form. If the user needs the recovered result in Hankel matrix form, simply apply $\G(\bz^{(K)})$ on the output vector $\bz^{(K)}$.
Therefore, we conclude HSGD is both computationally and memory efficient. In particular, the overall computational complexity is as low as $\cO(rn\log n+r^2n)$, as we discussed.


\begin{algorithm}[t]
\caption{\textbf{H}ankel \textbf{S}tructured \textbf{G}radient \textbf{D}escent (HSGD)} \label{algo:HSGD}
\begin{algorithmic}[1]
\State \textbf{Input:} 
$\Pi_\Omega\bf$: partial observation on the corrupted Hankel matrix in reweighted vector form; $r$: the rank of underlying Hankel matrix; $p$: observation rate; $\alpha$: outlier density; $\{\gamma_k\}$: parameters for sparsification operator.
\State \textcolor{OliveGreen}{// Initialization:}
\State $\bs^{(0)}=\D\Gamma_{\alpha p}(\D^{-1}\Pi_\Omega\bf)$
\State $[\BU^{(0)},\BSigma^{(0)},\BV^{(0)}]=\SVD_r(p^{-1}\G(\Pi_{\Omega}\bf-\bs^{(0)}))$
\State $\tilde{\bm{L}}^{(0)}=\BU^{(0)}(\BSigma^{(0)})^{1/2}$, \qquad $\tilde{\bm{R}}^{(0)}=\BV^{(0)}(\BSigma^{(0)})^{1/2}$
\State Define $\L$ and $\R$ by \eqref{eq:projectionsets}
\State $\BL^{(0)}=\Pi_{\L}\tilde{\bm{L}}^{(0)}, \qquad \ \BR^{(0)}=\Pi_{\R}\tilde{\bm{R}}^{(0)}$
\State \textcolor{OliveGreen}{// Iterative updates:}
\For{ $k=0,1,\ldots,K-1$ }
    \State $\bz^{(k)}=\G^*(\BL^{(k)}\BR^{(k)*})$
    \State $\bs^{(k+1)}=\Gamma_{\gamma_k\alpha p}(\Pi_{\Omega}(\bf-\bz^{(k)})))$
    \State $\BL^{(k+1)}=\Pi_\L(\BL^{(k)}-\eta\nabla_{\BL} \ell(\BL^{(k)},\BR^{(k)};\bs^{(k+1)}))$
    \State $\BR^{(k+1)}=\Pi_\R(\BR^{(k)}-\eta\nabla_{\BR} \ell(\BL^{(k)},\BR^{(k)};\bs^{(k+1)}))$
\EndFor
\State \textbf{Output:} $\bz^{(K)}=\G^*(\BL^{(K)}\BR^{(K)*})$: the recovered low-rank Hankel matrix in reweighted vector form.
\end{algorithmic}
\end{algorithm}

\subsection{Recovery guarantee}

In this section, we present the recovery guarantee of the proposed HSGD. Denote $\BL^{\natural}:=\BU^{\natural}\BSigma^{\natural \frac{1}{2}}$ and $\BR^{\natural}:=\BV^{\natural}\BSigma^{\natural \frac{1}{2}}$ where $\BU^{\natural}\bm{\Sigma}^{\natural}\BV^{\natural *}$ is the compact SVD of the underlying Hankel matrix $\G\bz^{\natural}$.
Consider the error that measures the distance between $(\BL^{(k)},\BR^{(k)})$ and $(\BL^\natural,\BR^\natural)$:
\begin{equation}\label{def-error}
d_k:=\mathrm{dist}\Big(\BL^{(k)},\BR^{(k)};\BL^{\natural},\BR^{\natural}\Big):=\min_{\BQ\in\mathbb{Q}_r}\sqrt{\| \BL^{(k)}-\BL^{\natural} \BQ \|_\fro^2+\| \BR^{(k)}-\BR^{\natural} \BQ \|_\fro^2}.
\end{equation}
Therein, $\mathbb{Q}_r$ denotes the set of $r\times r$ rotation matrices and the best rotation matrix $\BQ$ is used to align $(\BL^{(k)}, \BR^{(k)})$ and $(\BL^\natural,\BR^\natural)$ since the matrix factorization is not unique. Note that the standard reconstruction error (i.e., in Frobenius norm) is controlled by
\begin{equation} \label{eq:control dist}
    \|\H\bx^\natural - \BL^{(k)}\BR^{(k)*}\|_\fro^2 \leq \sigma_1^\natural d_k^2,
\end{equation}
provided $d_k^2 \leq\sigma_1^\natural$ \cite{yi2016fast}. Thus, it is sufficient to bound $d_k$ in our analysis.

We are ready to present our main results now. Firstly, we present the local linear convergence of HSGD as \cref{thm:convergence}, provided a sufficiently close initial guess.

\begin{theorem}\label{thm:convergence}
Suppose \cref{amp:Bernoulli,amp:incoherence,amp:sparsity} hold with $p\gtrsim \cO((c_s^2\mu^2 r^2\log n)/n)$ and  $\alpha\lesssim \cO(1/( c_s\mu r\kappa^2))$. Choose the parameters $\lambda=1/16$, $\gamma_k\in [1 + 1/b_0,2]$ with some fixed $b_0\geq 1$, and $\eta=\tilde{c}/\sigma_1^{\natural}$ with some sufficiently small $\tilde{c}$. If the initialization satisfies
$
    d_0 \lesssim \cO (\sqrt{\sigma_r^{\natural}/\kappa}),
$
then with probability at least $1-6n^{-2}$, the iterations of HSGD satisfy
\begin{equation*}
    d_{k+1}^2\le \Big(1-\frac{\eta \sigma_r^{\natural}}{64}\Big)^{k}d_0^2.
\end{equation*}
\end{theorem}
\begin{proof}
The proof of this theorem is deferred to \cref{subsec:prooflocalcon}.
\end{proof}


Next, we present the initialization guarantee as \cref{thm:Initialization}. Therein, we show the initial guess $(\BL^{(0)},\BR^{(0)})$ falls in the basin of attraction that specified in \cref{thm:convergence}.

\begin{theorem}\label{thm:Initialization}
Suppose\,\cref{amp:Bernoulli,amp:incoherence,amp:sparsity}\,hold with $p\geq (\varepsilon_0^{-2}\kappa^3 c_s\mu r^2\log n)/n$ and  $\alpha\leq 1/(32 c_s\mu r\kappa)$ where $\varepsilon_0\in (0,\frac{\sqrt{\kappa r}}{8c_0})$ with some constant $c_0$. Then with probability at least $1-2n^{-2}$, the initialization step of HSGD satisfies
\begin{equation*}
d_0\leq 26\alpha c_s\kappa \mu r\sqrt{r} \sqrt{\sigma_{r} ^{\natural}}+7c_0\varepsilon_0\sqrt{\sigma_{r} ^{\natural}/\kappa}.
\end{equation*}
\end{theorem}
\begin{proof}
The proof of this theorem is deferred to \cref{subsec:proofini}.
\end{proof}



Note that the probabilities in \cref{thm:convergence,thm:Initialization} come from some analogous restricted isometry properties (see \cref{RIP1,projectionerr1,projectionerror}), which hold uniformly for our results. Thus, all our theorems hold uniformly with a probability at least $1-6n^{-2}$. By directly combining \eqref{eq:control dist}, \cref{thm:convergence,thm:Initialization}, we show HSGD has global convergence to the ground truth with high probability, provided sufficiently many samples and sufficiently sparse outliers.

\begin{corollary}
Suppose \cref{amp:Bernoulli,amp:incoherence,amp:sparsity} hold with
\begin{equation*}
\begin{split}
\small
     p \gtrsim \cO\bigg(\frac{\max\{c_s^2\mu^2 r^2\log n, c_s\mu \kappa^3 r^2\log n\}}{n} \bigg)  \textnormal{ and }
     \alpha \lesssim \cO\bigg( \frac{1}{\max\{c_s \mu\kappa^{3/2}r^{3/2}, c_s \mu r\kappa^2\} }\bigg).
\end{split}
\end{equation*}
Then, HSGD finds an $\varepsilon$-optimal solution, i.e., $\|\H\bx^\natural - \BL^{(K)}\BR^{(K)*}\|_\fro/\|\H\bx^\natural\|_\fro \leq \varepsilon$,
in $K=\cO(\kappa \log \varepsilon^{-1})$ iterations with probability at least $1-6n^{-2}$.
\end{corollary}

\begin{remark}
When outliers are not appearing (i.e., $\alpha=0$), the RHC problem reduces to the vanilla Hankel matrix completion problem, and HSGD becomes projected gradient descent (PGD) introduced in \cite{cai2018spectral}. In the original paper, PGD theoretically requires $\cO(\kappa^2 \log \varepsilon^{-1})$ iterations to find a $\varepsilon$-optimal solution. With the improved proof techniques, we show that running $\cO(\kappa \log \varepsilon^{-1})$ iterations is sufficient for HSGD. Thus, as a special case of HSGD, we also theoretically improve the convergence speed of PGD to $\cO(\kappa \log \varepsilon^{-1})$ iterations.
\end{remark}

\begin{remark}
The convergence rate of HSGD suggests that the proposed algorithm runs faster on well-conditioned problems, i.e., $\kappa=\cO(1)$. In many applications, the condition number of the underlying Hankel matrix is indeed good. For example, in a NMR spectroscopy problem, $\kappa$ depends on the ratio between the largest and smallest magnitudes of the complex amplitudes \cite[Remark~1]{cai2019fast}, which typically is modest.
\end{remark}

\section{Numerical experiments}
In this section, we compare the proposed HSGD against the state-of-the-art RHC approaches, PartialSAP \cite{zhang2019correction} and RobustEMaC \cite{chen2014robust}. We demonstrate the empirical advantages of HSGD on both synthetic and real datasets. We hand tuned the parameters for their best performance. In particular, we use a iterative decaying $\gamma_k= 1.05+0.45\cdot 0.95^{k}$ for HSGD, so it starts with $\gamma_0=1.5$ and $\gamma_k \rightarrow 1.05$ as $k\rightarrow\infty$. By \cref{thm:convergence}, any $\gamma_k\in[1.05,2]$ (i.e., $b_0=20$) will work, we find the iterative decaying $\gamma_k$ provides the best empirical performance for HSGD. The reason behind this parameter choice is HSGD gets better outlier estimations in the latter iterations, so less amount of false-positive outliers will be taken. All numerical experiments were performed from Matlab on a Windows laptop equipped with Intel i7-8750H CPU and 32GB RAM. For a fair comparison, PROPACK \cite{larsen2004propack} was used for fast truncated SVD in all tested algorithms.
The Matlab implementation of HSGD is available online at \url{https://github.com/caesarcai/HSGD}.

\subsection{Synthetic examples}

We generate the rank-$r$ Hankel matrices via two steps: (i) generate a vector $\bx^\natural\in\C^n$ that is sparse in Fourier space with exact $r$ active frequencies;\footnote{We follow the same method used in \cite[section~III.A]{cai2021asap} to generate such vectors. In our tests, we ensure the active frequencies are well separated in the generated vectors.}
then (ii) generate the corresponding Hankel matrix $\H(\bx^\natural)\in\C^{n_1\times n_2}$ with $n_1\approx n_2\approx n/2$.\footnote{If $n$ is odd, we use $n_1=n_2=(n+1)/2$. If $n$ is even, we use $n_1=n_2-1=n/2$.}
Such a Hankel matrix must be rank-$r$
\cite{liao2016music}. For the suitable algorithms, we also generate the reweighted vector $\bz^\natural=\D\bx^\natural$ so that $\H(\bx^\natural)=\G(\bz^\natural)$.
We uniformly (without replacement) observe $m:=pn$ entries from $\bz^\natural$, then we uniformly choose $\alpha m$ entries among the observed ones to be corrupted. The corruption is done by adding complex outliers whose real parts and imaginary parts are drawn \textit{i.i.d.}~from the uniform distribution over the intervals $[-10\mathbb{E}(|\mathrm{Re}(z_i^\natural)|),10\mathbb{E}(|\mathrm{Re}(z_i^\natural)|)]$ and $[-10\mathbb{E}(|\mathrm{Im}(z_i^\natural)|),10\mathbb{E}(|\mathrm{Im}(z_i^\natural)|)]$, respectively. In the experiments, we use a uniform sampling model instead of Bernoulli sampling model since the former is easier to control the number of samples and outliers. We emphasize that the empirical behaviors are not much different between these two sampling models.

\vspace{0.05in}
\textbf{Empirical phase transition.}
In this section, we present the recoverability of the tested algorithms under various settings. We fix $n=125$ for all experiments in the section.
The pixels on the phase transition plots represent different problem settings. For each pixel, we conduct $50$ testing problems, then a white pixel means all $50$ cases were recovered and a black pixel means all $50$ cases were failed. Specifically speaking, the output of a testing problem is considered a successful recovery if $\|\G\bz^{(K)}-\G\bz^\natural\|_\fro/\|\G\bz^\natural\|_\fro\leq10^{-3}$ while the stopping criteria is $\|\G\Pi_{\Omega}\bz^{(k)}+\G\Pi_{\Omega}\bs^{(k)}-\G\Pi_{\Omega}\bf \|_\fro/\|\G\Pi_{\Omega}\bf\|_\fro\leq 10^{-5}$.\footnote{When suitable, we calculate the empirical residue in the equivalent vector form to save runtime.}

\begin{figure}[t]
    \centering
    \includegraphics[width = 0.30\linewidth]{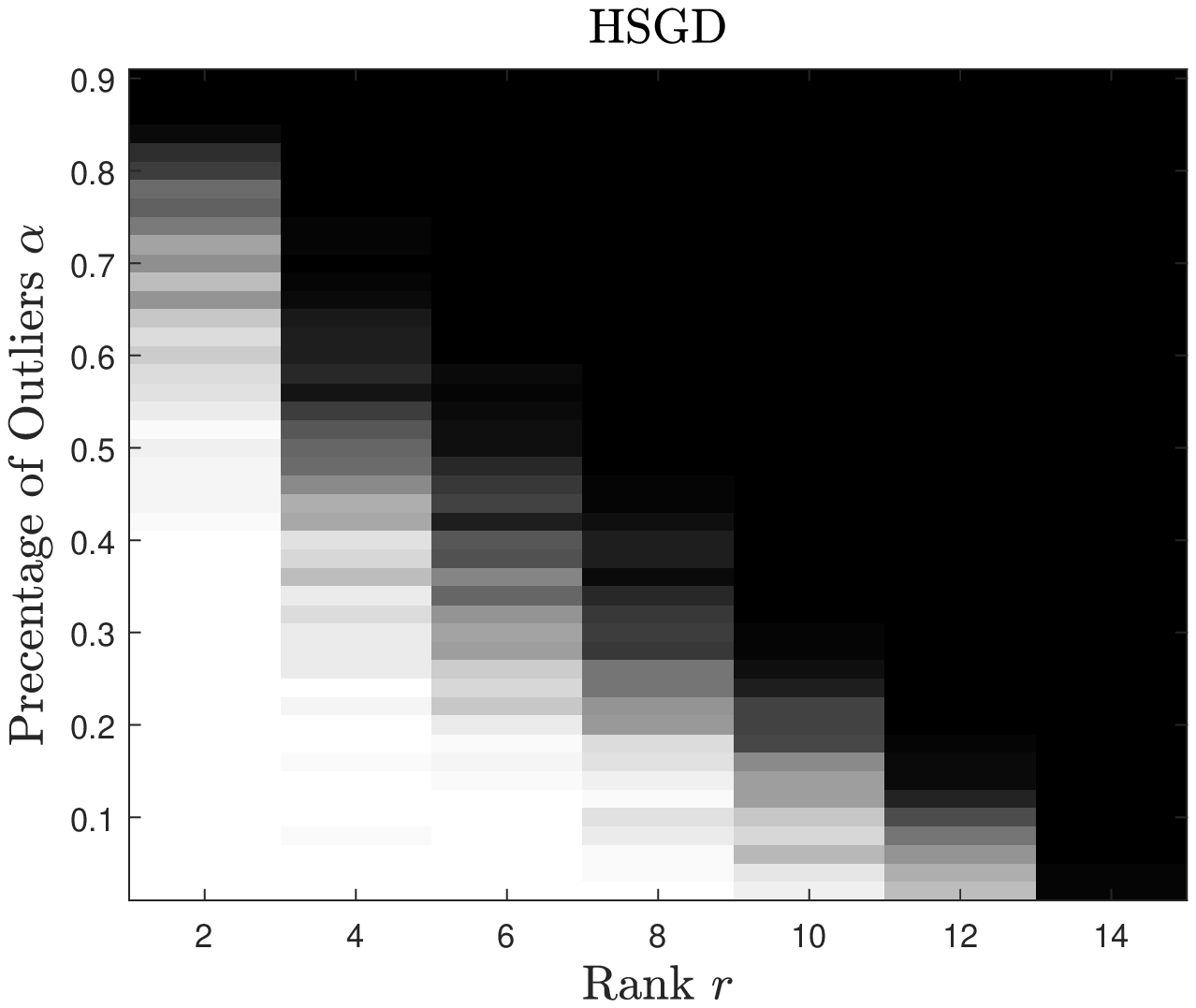}
    \hfill
    \includegraphics[width = 0.30\linewidth]{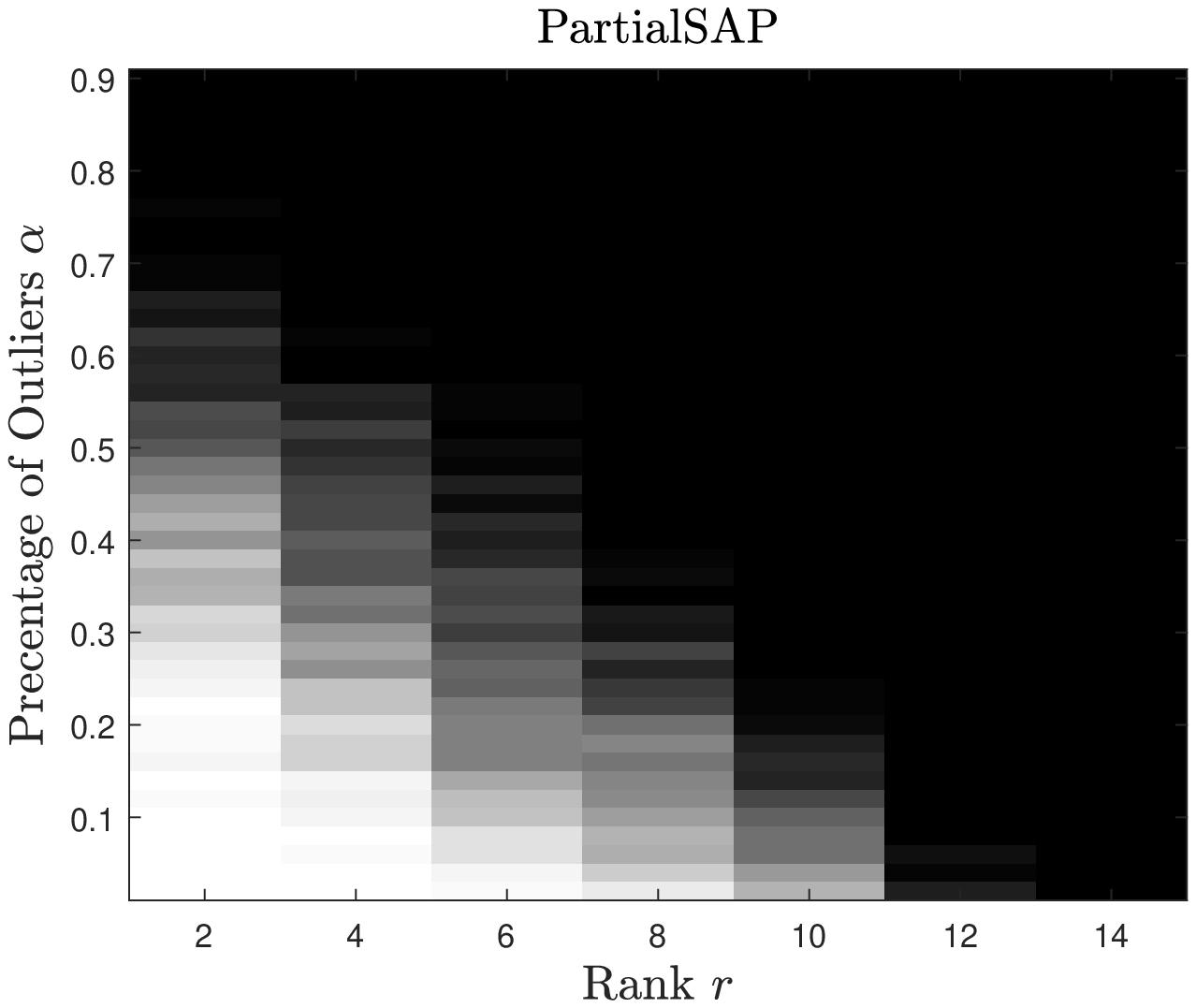}
    \hfill
    \includegraphics[width = 0.30\linewidth]{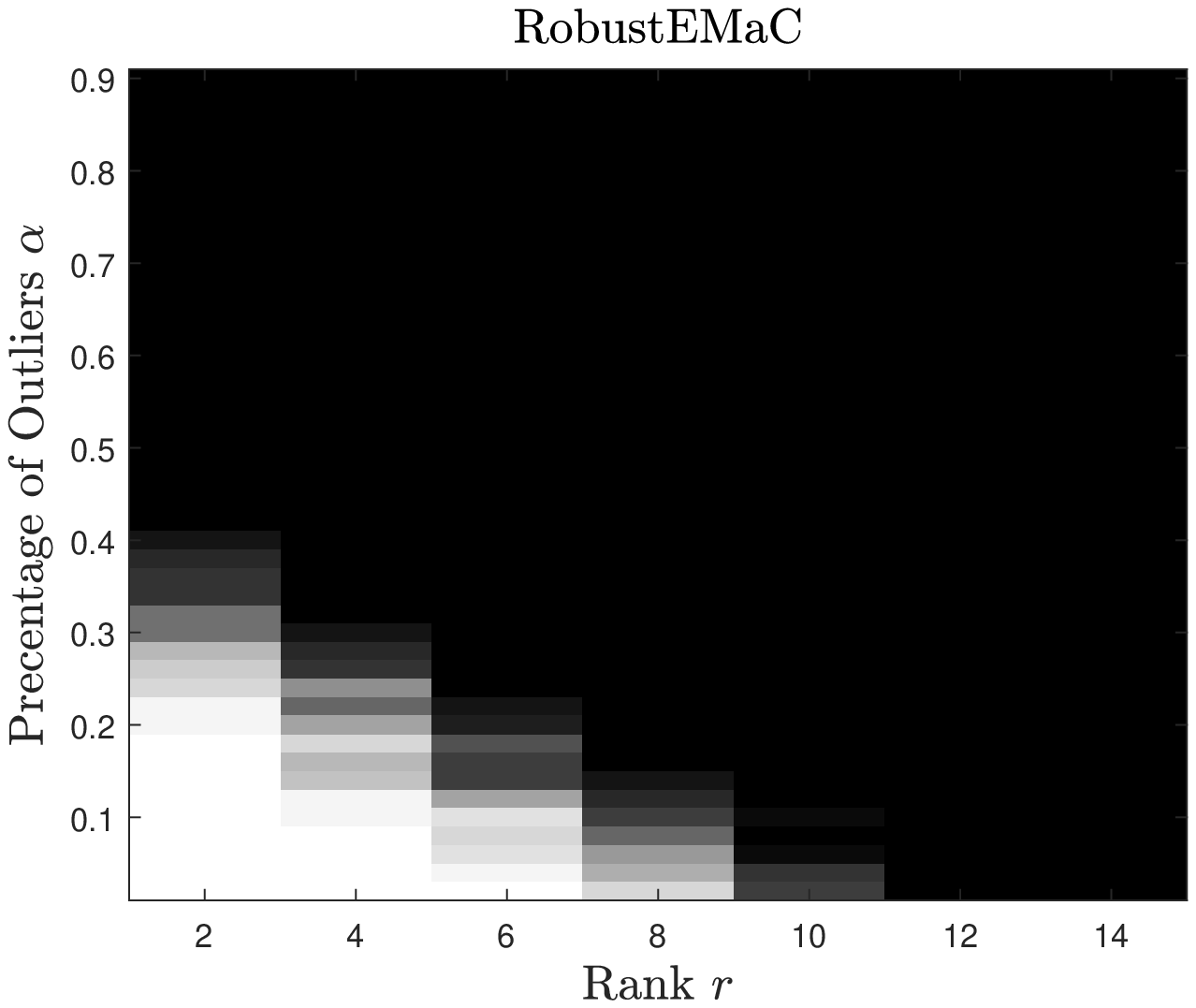}
    \caption{Empirical phase transition for HSGD, PartialSAP, and RobustEMaC: Rank \textit{vs.} rate of outliers. $50$ entries are sampled in all testing problems.}
    \label{fig:rank_vs_outlier}
\end{figure}

\begin{figure}[t]
    \centering
    \includegraphics[width = 0.30\linewidth]{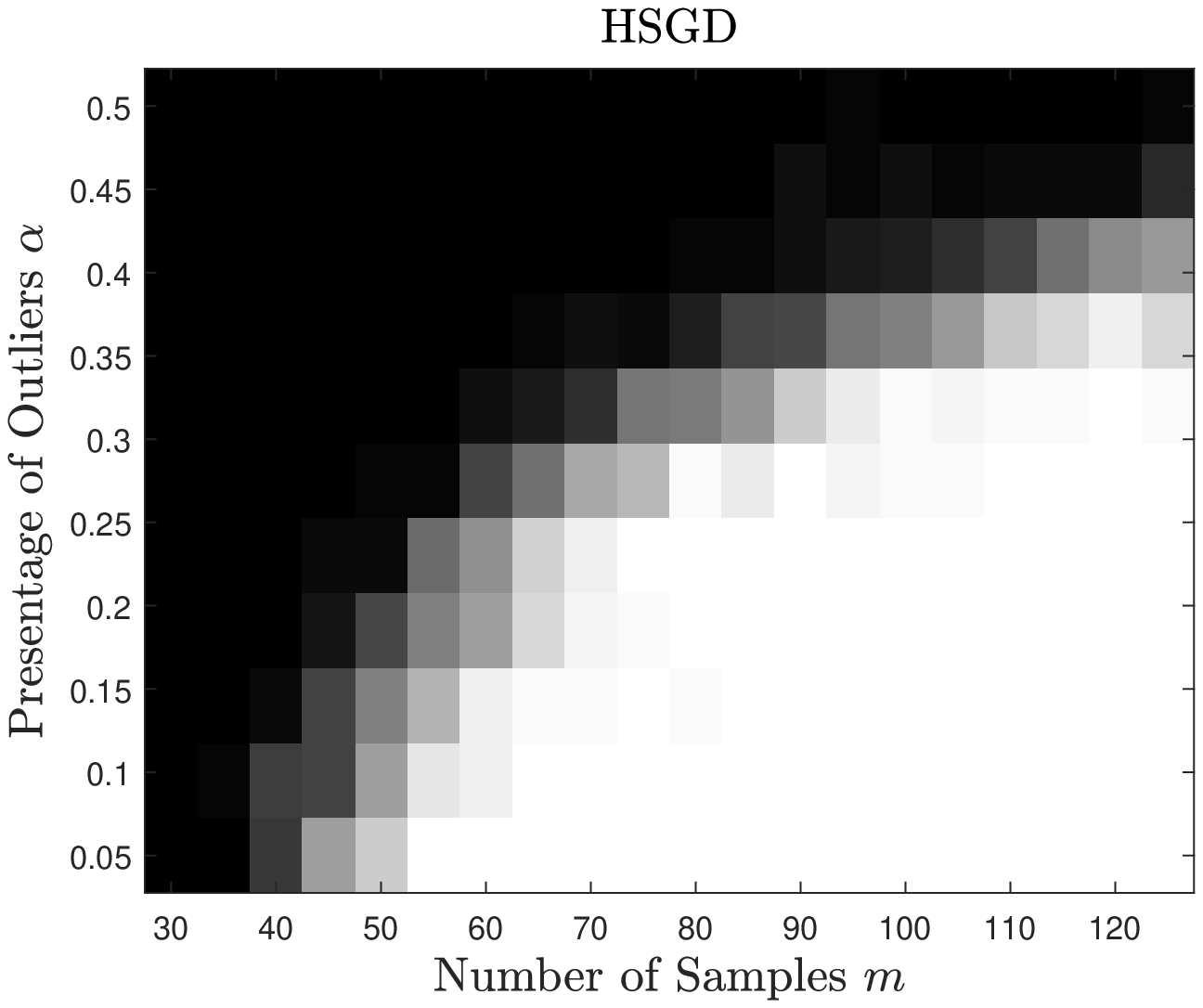}
    \hfill
    \includegraphics[width = 0.30\linewidth]{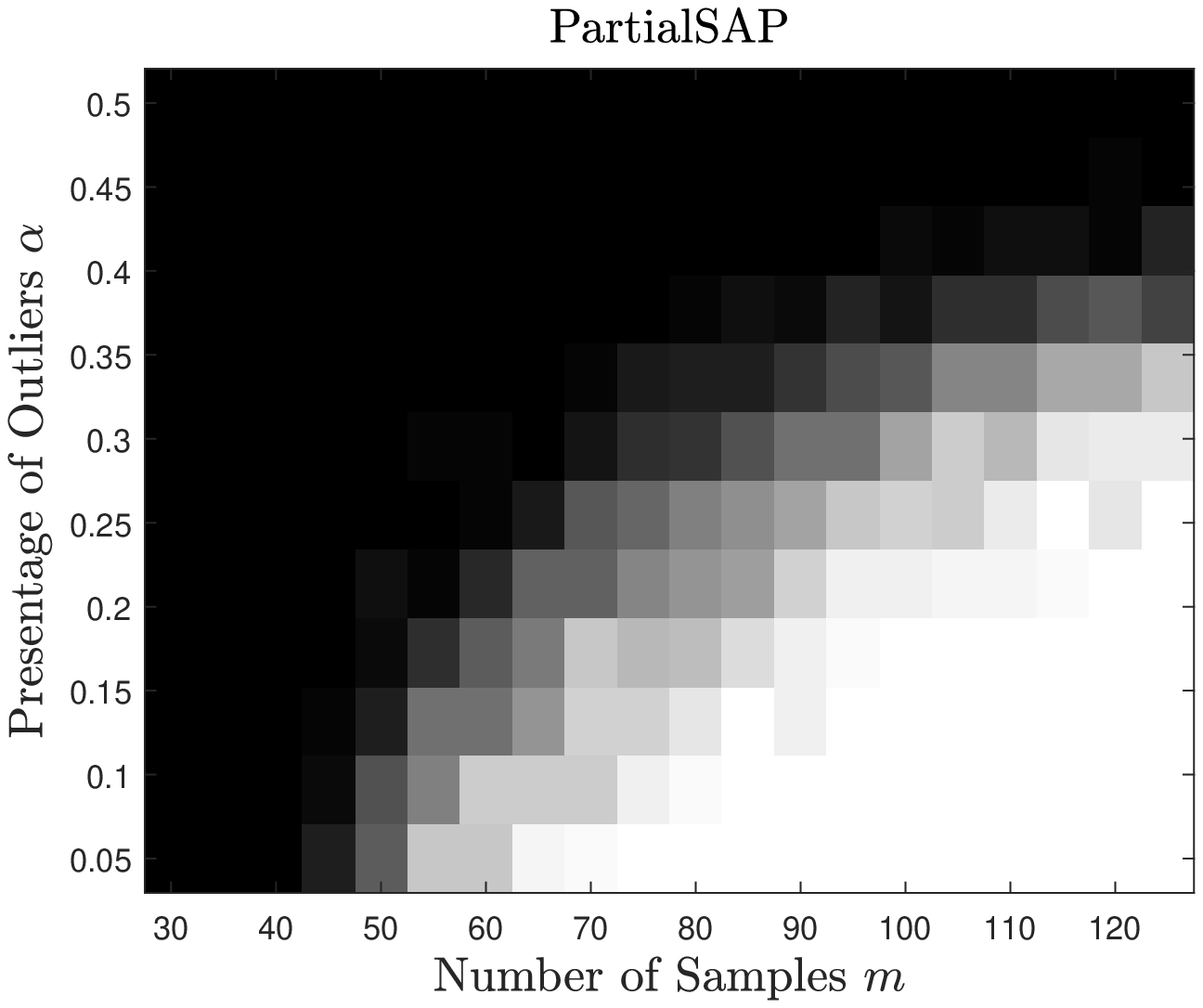}
    \hfill
    \includegraphics[width = 0.30\linewidth]{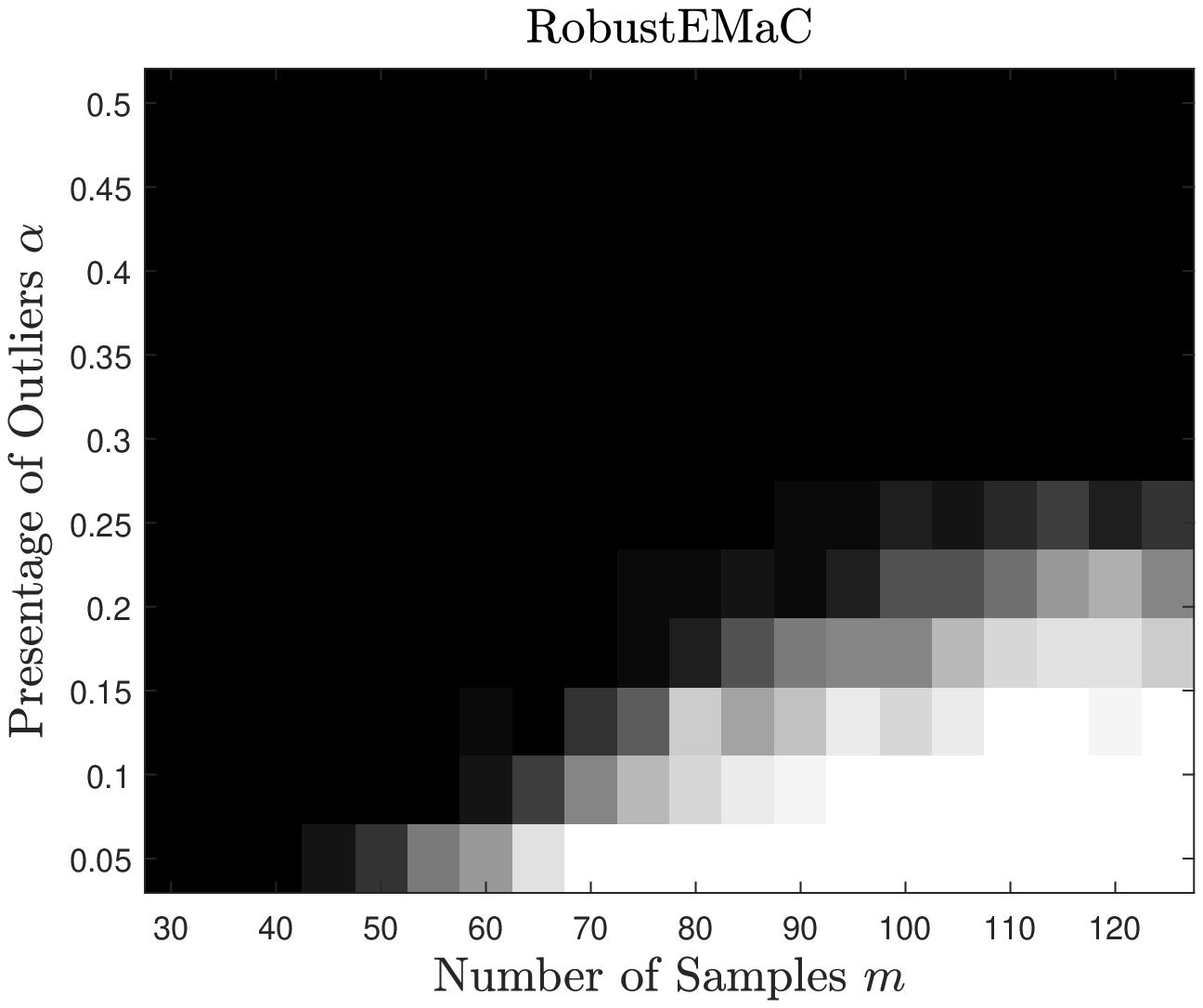}
    \caption{Empirical phase transition for HSGD, PartialSAP, and RobustEMaC: Number of samples \textit{vs.} rate of outliers. All testing problems have rank $10$.}
    \label{fig:sample_rate_vs_outlier}
\end{figure}

\begin{figure}[t]
    \centering
    \includegraphics[width = 0.30\linewidth]{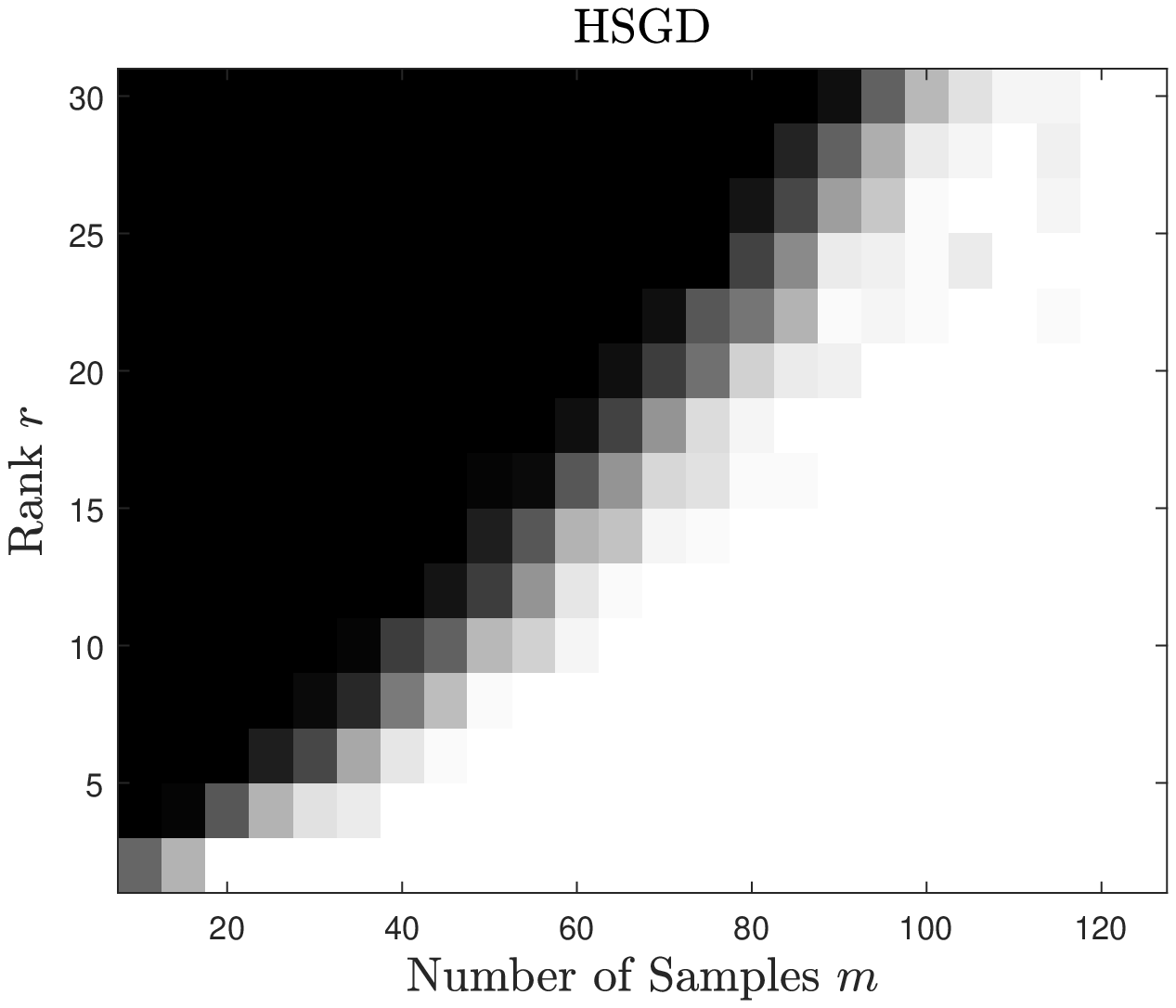}
    \hfill
    \includegraphics[width = 0.30\linewidth]{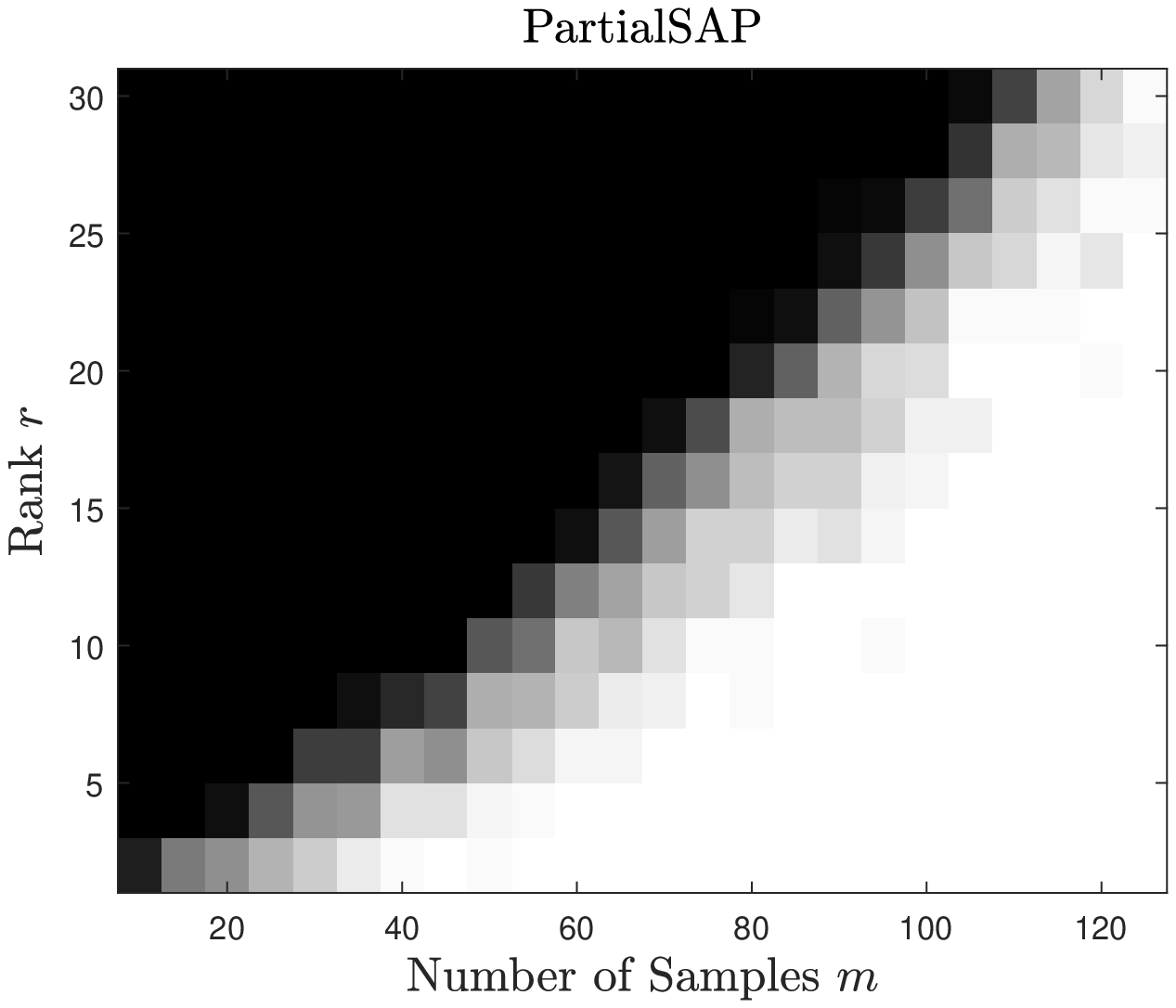}
    \hfill
    \includegraphics[width = 0.30\linewidth]{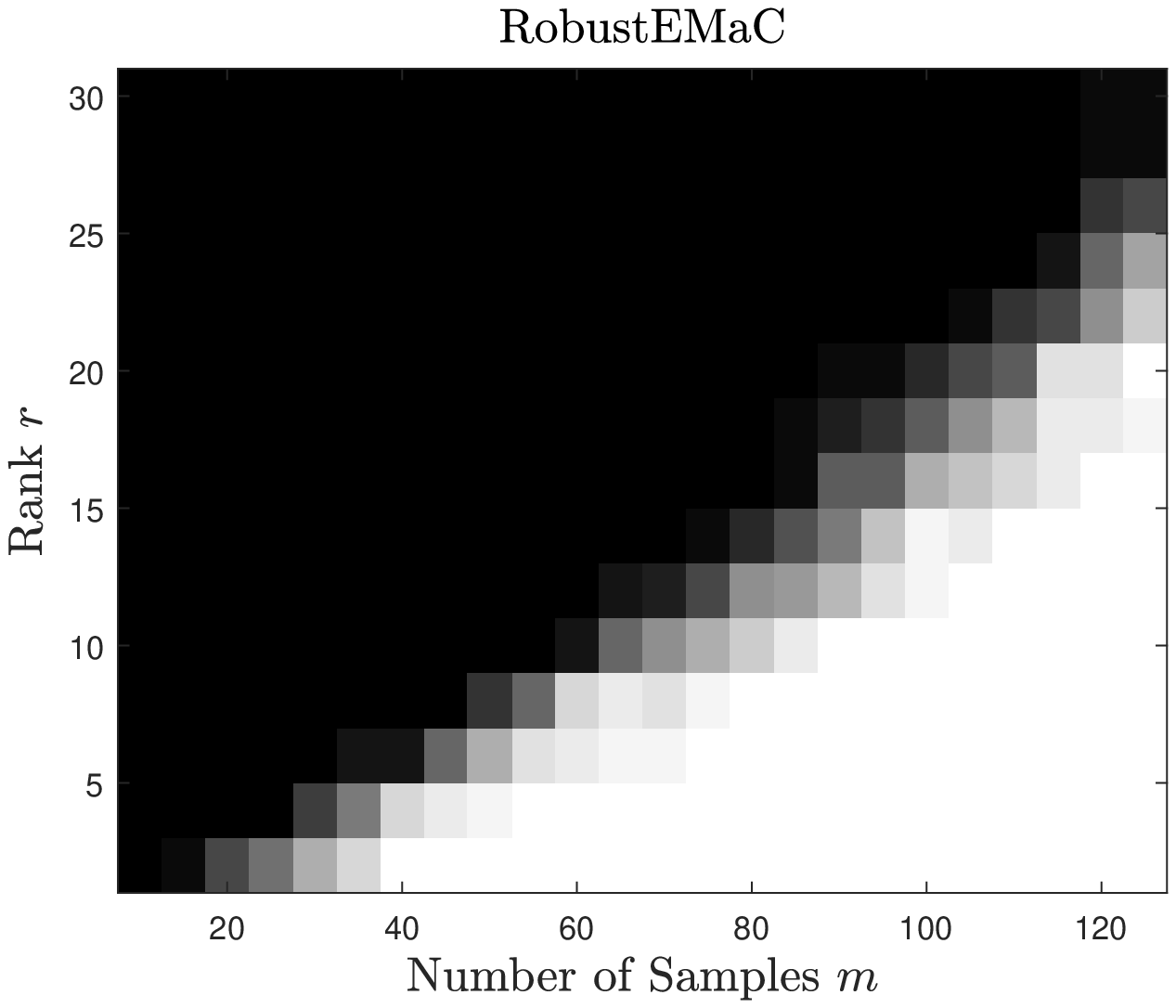}
    \caption{Empirical phase transition for HSGD, PartialSAP, and RobustEMaC: Number of samples \textit{vs.} rank. $10\%$ of samples are corrupted by outliers in all testing problems.}
    \label{fig:rank_vs_sample_rate}
\end{figure}

In \cref{fig:rank_vs_outlier}, we fix the number of samples $m=50$ and study the empirical phase transition with varying rank $r$ and outlier sparsity $\alpha$.
In \cref{fig:sample_rate_vs_outlier}, we fix rank $r=10$ and study the empirical phase transition with varying outlier sparsity $\alpha$ and number of samples $m$.
In \cref{fig:rank_vs_sample_rate}, we fix the outlier sparsity $\alpha=0.1$ and study the empirical phase transition with varying rank $r$ and number of samples $m$. In all three comparisons, we find HSGD has the best recoverability and the other non-convex algorithm, i.e., PartialSAP, is competitive.

\vspace{0.05in}
\textbf{Computational efficiency.}
We demonstrate the computational efficiency of the tested algorithms. All testing problems is this experiment have rank $r=10$, observation rate $p=40\%$, and outlier sparsity $\alpha=10\%$. The reported runtime is averaged over 20 trials. Since the large-scale problems are prohibitive for the convex method Robust-EMaC, so we only compare HSGD against PartialSAP in this section. While both non-convex algorithms have similar complexity orders, the leading constant for HSGD is expected to be much smaller. In the left subfigure of \cref{fig:speed_tests}, we exponentially increase the problem dimension $n$ and record the runtime. We observe that HSGD is $10\times$ faster than PartialSAP when the problem dimension is large. In the middle subfigure of \cref{fig:speed_tests}, we run the same dimension vs. runtime experiment with only HSGD and even larger dimensions. In this plot, we match the logarithmic base for $x$- and $y$-axis in the log-log plot and also include error bars. The slope of this log-log plot is approximately $1$ when $n$ is large, this verifies the claimed computational complexity for HSGD, i.e., the dependence on problem dimension is merely $\cO(n\log n)$. Moreover, the narrow error bar in the plot shows the runtime of HSGD is stable. In the right subfigure of \cref{fig:speed_tests}, we present the convergence behavior of the tested algorithms where we fix $n=2^{12}$. One can see both algorithms have linear convergence as the theorems indicated, and HSGD has a more sharp convergence rate with respect to runtime. Overall, we conclude HSGD is a highly efficient algorithm, compared to the state-of-the-art RHC approaches.

\begin{figure}[t]
    \centering
    \includegraphics[width = 0.3\linewidth]{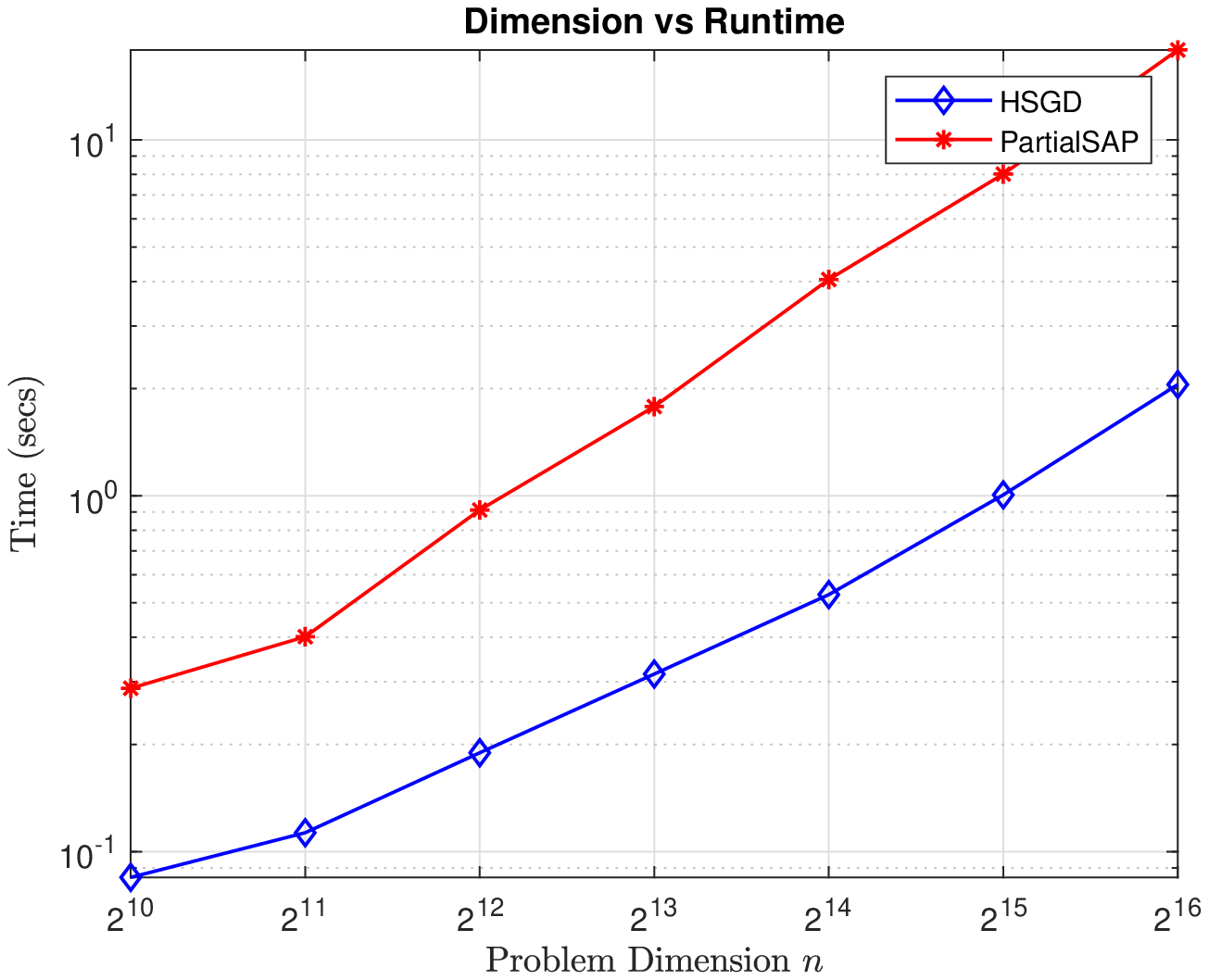}
    \hfill
    \includegraphics[width = 0.3\linewidth]{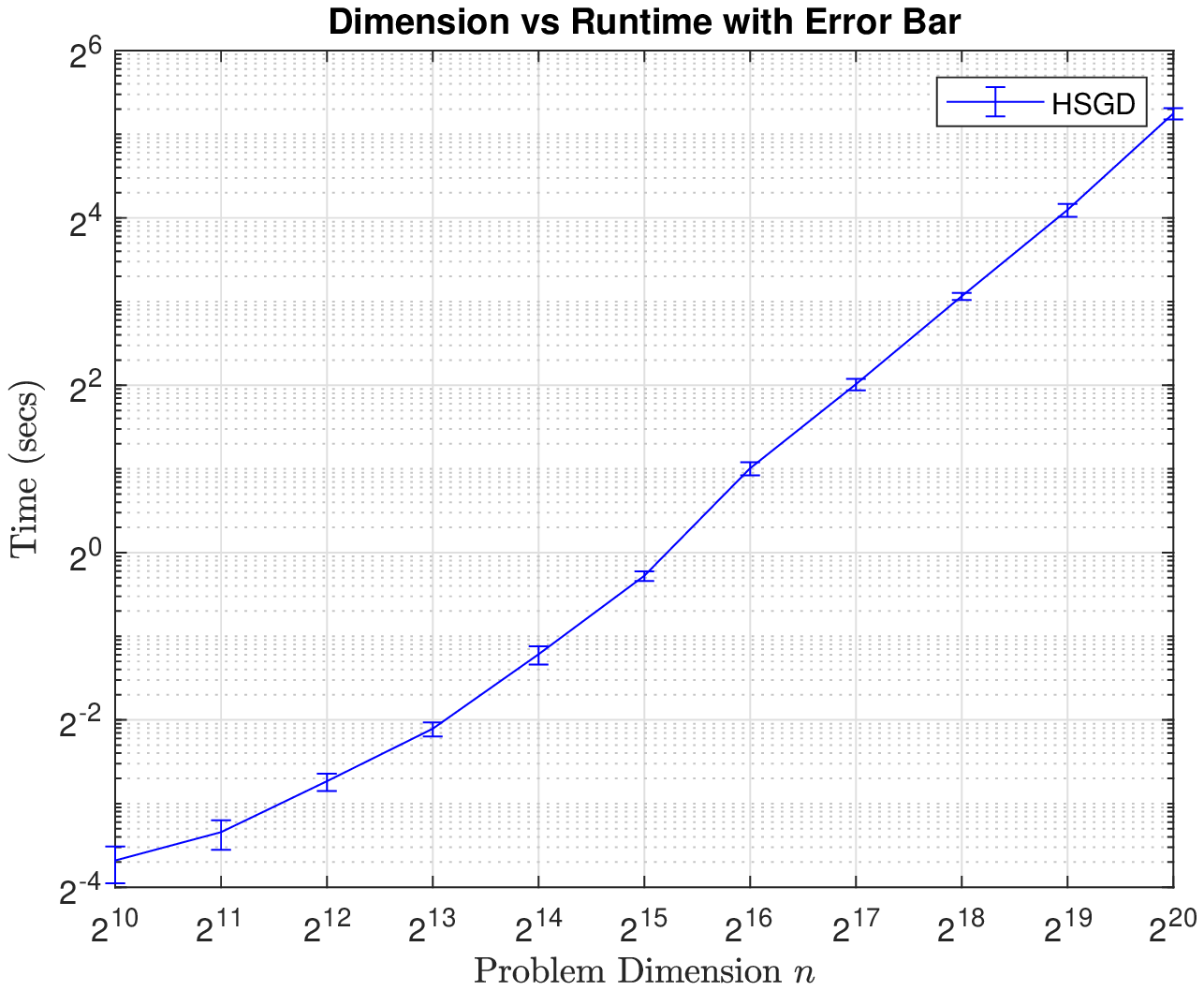}
    \hfill
    \includegraphics[width = 0.3\linewidth]{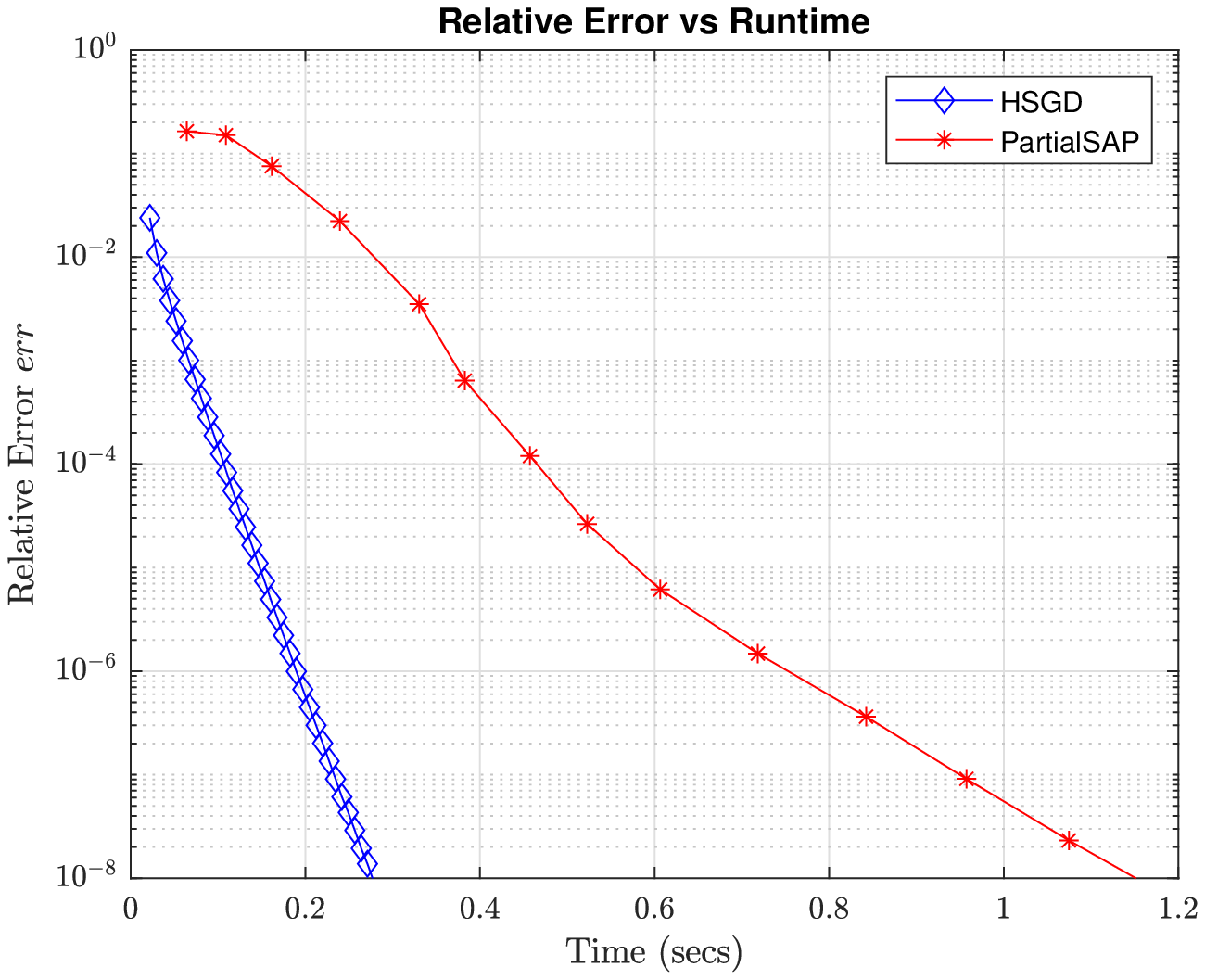}
    \caption{Experimental results for speed tests between HSGD and PartialSAP. \textbf{Left:} Dimension \textit{vs.} runtime. \textbf{Middle:} Dimension \textit{vs.} runtime with error bar for HSGD only. \textbf{Right:} Relative error \textit{vs.} runtime. }
    \label{fig:speed_tests}
\end{figure}
\subsection{Nuclear magnetic resonance spectroscopy}
As we discussed in \cref{sec:introduction}, NMR signal recovery is a widely used real-world benchmark for the problems of low-rank Hankel matrix: Given a clear one-dimensional NMR signal $\bx^\natural$, the corresponding Hankel matrix $\H\bx^\natural$ is rank-$r$ where $r$ is determined by the number of bars in the power spectrum of the signal.
In this section, we apply RHC algorithms to complete the partially observed NMR signal and remove the impulse corruptions, simultaneously. The testing NMR signal has the dimension $n=32,768$ and rank $r\approx 40$, which is prohibitively large size for RobustEMaC. We test HSGD and PartialSAP for recovering this signal under various observation rate $p$ and outlier sparsity $m$.
The runtime comparison results are summarized as \cref{tab:comtimeNMR} where tested algorithms recover the desired signal in all cases. One can see that HSGD maintains his speed advantage in this real-world application, under each of the settings.
Moreover, in \cref{fig:NMR}, we demonstrate the power spectrum of the signal recovered by HSGD. Therein, we not only successfully recovered the NMR signal but also clear the small noise in the original data.
Although it is not theoretically verified, the empirical results suggests that HSGD can also denoise small white noise when it detects extreme outliers.

\begin{table}[t]
\centering
\begin{small}
\caption{Runtime comparison between HSGD and PartialSAP for NMR signal recovery under various observation rate $p$ and outlier sparsity $\alpha$.}\label{tab:comtimeNMR}
\begin{tabular}{c|c c c c c}
\toprule
 ($p$, $\alpha$)  &($0.3$,$0.1$)   &($0.3$, $0.2$) &($0.3$, $0.3$)   &($0.4$, $0.4$)   &($0.5$, $0.5$)   \\
\midrule
PartialSAP   & $335.09$\,s   &$346.328$\,s  &$352.01$\,s     &$357.75$\,s   &$369.32$\,s   \\
HSGD  & $31.142$\,s   & $33.67$\,s  &$45.82$\,s    & $53.69$\,s   &$64.11$\,s  \\
\bottomrule
\end{tabular}
\end{small}
\end{table}

\begin{figure}[t]
\label{fig:NMR}
  \centering
    \includegraphics[clip=true,width = 0.45\linewidth]{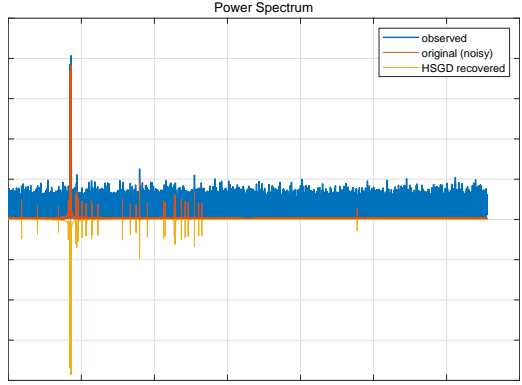}
    \caption{Power spectrum of the noisy original NMR signal, the observed signal ($p=30\%$ and $\alpha=30\%$), and HSGD recovered signal (upside down). Note that the observed signal in picture is rescaled by $1/p$, which is a common method to offset the energy loss due to partial observation.}
\end{figure}

\section{Conclusion remarks}\label{section:conclusion}
In this work, we proposed a novel non-convex algorithm, coined Hankel Structured Gradient Descent (HSGD), for robust Hankel matrix completion problems. HSGD is highly computing- and sample-efficient. In particular, HSGD costs merely $\cO(r^2n+rn\log n)$ flops per iteration while it requires as few as $\cO(\max\{c_s^2\mu^2r^2 \log n, c_s\mu \kappa^3r^2 \log n\})$ samples. HSGD is also robust and tolerates $\cO(1/\max\{c_s \mu\kappa^{3/2}r^{3/2}, c_s \mu r\kappa^2\} )$-fraction outliers. Theoretical recovery guarantees have been established for HSGD, along with a provable linear convergence rate. The superior performance of HSGD, in terms of efficiency and robustness, is verified by numerical experiments on both synthetic and real datasets.

\appendix
\section{Proofs of theoretical results} \label{sec:proofs}
In this section, we provide the analysis for the claimed theoretical results.
All the proofs are under \cref{amp:Bernoulli,amp:incoherence,amp:sparsity}. We start with introducing some addition notation used in the analysis.
We denote the tangent space of rank-$r$ matrix manifold at $\G \bm{z}^{\natural}$ by
\begin{equation*}
T:=\left\{\bm{X}~|~ \bm{X}=\bm{U}^{\natural} \bm{C}^*+\bm{D} \bm{V}^{\natural *}\ \text{where}\ \bm{C}\in\mathbb{C}^{n_1\times r}, \bm{D}\in\mathbb{C}^{n_2\times r}\right\}.
\end{equation*}
For any $\bm{Q}\in \mathbb{Q}_r$, $\left(\bm{L}^{\natural}\bm{Q},\bm{R}^{\natural}\bm{Q}\right)$ is an equivalent solution to $\left(\bm{L}^{\natural},\bm{R}^{\natural}\right)$. Thus, we define the solution set to be
\begin{equation*}
\begin{aligned}
 \E\left(\bm{L}^{\natural},\bm{R}^{\natural}\right)
 :=&\{\left(\bm{L},\bm{R}\right)\in \mathbb{C}^{n_1\times r}\times \mathbb{C}^{n_2\times r}~|~\bm{L}=\bm{L}^{\natural}\bm{Q},\bm{R}=\bm{R}^{\natural}\bm{Q}
 ~\text{where}~ \bm{Q}\in \mathbb{Q}_r\}.
\end{aligned}
\end{equation*}
 The sequence $\{\tilde{d}_k\}_{k\ge0}$ is defined as
$$
\tilde{d}_{k}:=\mathrm{dist}(\tilde{\bm{L}}^{(k)},\tilde{\bm{R}}^{(k)};\bm{L}^{\natural},\bm{R}^{\natural}),
$$
where we use
$\tilde{\bm{L}}^{(k)}:=\BL^{(k-1)}-\eta\nabla_{\BL} \ell(\BL^{(k-1)},\BR^{(k-1)};\bs^{(k)})$ and $\tilde{\bm{R}}^{(k)}:=\BR^{(k-1)}-\eta\nabla_{\BR} \ell(\BL^{(k-1)},\BR^{(k-1)};\bs^{(k)})$ to denote the middle step of \eqref{eq:update_Hankel} before the projection onto $\L$ and $\R$ respectively.
With a slight abuse of notation, we let $\P_{\Omega}$ be the projection operator onto the space which can be represented by an orthonormal basis of Hankel matrices. That is, for any matrix $\bm{Z}\in \mathbb{C}^{n_1\times n_2}$,
\begin{equation*}
\P_{\Omega}\left(\bm{Z}\right)=\sum_{a\in\Omega}\l \bm{Z}, \bm{H}_a\r \bm{H}_a,
\end{equation*}
where $\{\bm{H}_a\}_{a=1}^n$ is the orthonormal basis of Hankel matrices, defined by
\begin{equation*}
\bm{H}_a:=\textstyle{\frac{1}{\sqrt{\varsigma_a}}}\H \bm{e}_a,  \quad \textnormal{where } \bm{e}_a~\text{is the} ~a\text{-th standard basis vector of } \mathbb{R}^n.
\end{equation*}
By this definition, for any Hankel matrix $\G\bm{z}$, we have $\P_{\Omega}\G\bm{z}=\G\Pi_{\Omega}\bm{z}$. $\delta_a$ is defined as
 \begin{equation}  \label{eq:delta_a}
 \delta_a=
 \begin{cases}
 1,\quad &\textnormal{with probability}\ p; \cr
 0,      &\textnormal{otherwise},
 \end{cases}.
 \end{equation}
 for all $a\in [n]$.

\subsection{Technical lemmas} \label{sec:tech_lemma}
We prove some technical lemmas that will be used in the convergence analysis in this subsection.
\begin{lemma}\label{projectionerror}
 There exists some universal constant $c_0>0$ such that
\begin{equation}\label{event:RIP3}
\lV \left(p^{-1}\G \Pi_\Omega -\G\right) \bm{z}^{\natural}\rV_2\le c_0\sqrt{(p^{-1}\mu c_s r\log n)/n}\lV \G \bm{z}^{\natural}\rV_2
\end{equation}
holds with probability at least $1-2n^{-2}$ provided $p\ge \left(\mu c_s r\log n\right)/n$.
\end{lemma}
\begin{proof}
By the definition of $\delta_a$ in \cref{eq:delta_a}, first we have
 \begin{equation*}
   p^{-1}\G \Pi_\Omega \bm{z}^{\natural}-\G \bm{z}^{\natural}=p^{-1}\left(\delta_a-p\right)\bm{z}^{\natural}_a \bm{H}_a.
 \end{equation*}
Denote $\bm{\Z}_a:=p^{-1}\left(\delta_a-p\right) \bm{z}^{\natural}_a\bm{H}_a$. Thus, $\mathbb{E}[\bm{\Z}_a]=0$. By $\lV \bm{H}_a \rV_2 \le \frac{1}{\sqrt{\varsigma_a}}$, we have
\begin{equation*}
\lV \bm{\Z}_a\rV_2\le p^{-1} \lv\bm{z}^{\natural}_a\lv \lV\bm{H}_a \rV_2 \le p^{-1} \frac{\lv \bm{z}^{\natural}_a \lv}{\sqrt{\varsigma_a}}\le p^{-1}\lV \D^{-1}\bm{z}^{\natural}\rV_{\infty}.
\end{equation*}
Note that $\bm{\Z}_a \bm{\Z}_a^*=\left(\textstyle{\frac{\delta_a}{p}}-1\right)^2\lv\bm{z}^{\natural}_a\lv^2\bm{H}_a\bm{H}_a^*$. Thus,
\begin{equation*}
\begin{aligned}
\Big\| \mathbb{E}\sum_a\bm{\Z}_a \bm{\Z}_a^*\Big\|_2
&\leq {\frac{1}{p}}\Big\|\sum_a\lv\bm{z}^{\natural}_a\lv^2\bm{H}_a\bm{H}_a^*\Big\|_2
\le {\frac{1}{p}} \big\|\mathrm{diag}\big(\G\bm{z}^{\natural}\left(\G\bm{z}^{\natural}\right)^*\big)\big\|_2
\le {\frac{1}{p}}\lV\G\bm{z}^{\natural}\rV_{2,\infty}^2.
\end{aligned}
\end{equation*}
Similarly, we have $\| \mathbb{E}\left(\sum_a \bm{\Z}_a^*\bm{\Z}_a\right)\|_2\le \frac{1}{p}\|(\G\bm{z}^{\natural})^*\|_{2,\infty}^2$. Moreover, by the $\mu$-incoherence condition, we have
\begin{equation*}
\lV\G\bm{z}^{\natural}\rV_{2,\infty}^2=\mathop{\max}_i\lV \bm{e}_i^* \bm{U}^{\natural}\bm{\Sigma}^{\natural}\bm{V}^{\natural*}\rV_2^2\le \lV\G\bm{z}^{\natural}\rV_{2}^2\mathop{\max}_i\lV \bm{e}_i\bm{U}^{\natural}\rV\le \frac{c_s\mu r}{n} \lV\G\bm{z}^{\natural}\rV_{2}^2.
\end{equation*}
Similarly, we also have $\|(\G\bm{z}^{\natural})^*\|_{2,\infty}^2\le \frac{c_s\mu r}{n} \|\G\bm{z}^{\natural}\|_{2}^2$. Note that
\begin{equation*}
\|\D^{-1}\bm{z}^{\natural}\|_{\infty}=\lV\G \bm{z}^{\natural}\rV_{\infty}=\mathop{\max}_{i,j}\lv e_i^*\left(\G \bm{z}^{\natural}\right)e_j\lv\le \frac{c_s\mu r}{n} \lV\G\bm{z}^{\natural}\rV_{2}.
\end{equation*}
Using the bounds of $\|\G\bm{z}^{\natural}\|_{2,\infty}$, $\|(\G\bm{z}^{\natural})^*\|_{2,\infty}$ and $\|\D^{-1}\bm{z}^{\natural}\|_{\infty}$, the Bernstein's inequality \cite[Theorem~1.6]{tropp2012user} then yields
\begin{equation}\label{sumz}
\mathbb{P}\Big(\big\| \sum_a \bm{\Z}_a\big\|_2>t\Big)\le \left(n_1+n_2\right)\exp\left(\frac{-n p t^2/2}{ \lV\G\bm{z}^{\natural}\rV_{2}^2 c_s\mu r+t\lV\G\bm{z}^{\natural}\rV_{2}c_s\mu r/3}\right).
\end{equation}
Let $t=c_0\sqrt{\left(\mu c_s r\log n\right)/(pn)}\lV \G \bm{z}^{\natural}\rV_2$ and $p\ge \left(\mu c_s r\log n\right)/n$. We then have
\begin{equation*}
\mathbb{P}\Big(\lV p^{-1}\G \Pi_\Omega \bm{z}^{\natural}-\G \bm{z}^{\natural}\rV_2> c_0\sqrt{p^{-1}\mu c_s r\log n/n}\lV \G \bm{z}^{\natural}\rV_2\Big)\le 2n^{-2},
\end{equation*}
for any constant $c_0\ge4$.
\end{proof}

\begin{lemma}\label{RIP1}
 There exists a constant $c_4$ such that if $p\ge(\varepsilon_0^{-2}c_4 \mu r\log n)/n$, it holds
\begin{equation}\label{eq:projectionRIP1}
\lV p^{-1}\P_T\P_{\Omega}\G\G^*\P_T-\P_T \G\G^*\P_T \rV\le \varepsilon_0
\end{equation}
with probability at least $1-2n^{-2}$.
\end{lemma}
\begin{proof}
Notice $\G\G^*\bm{H}_a=\bm{H}_a$. For any matrix $\bm{X}\in \mathbb{C}^{n_1\times n_2}$, we have
\begin{equation*}
\begin{split}
\left(\textstyle{\frac{1}{p}}\P_T\P_{\Omega}\G\G^*\P_T-\P_T\G\G^*\P_T\right)\left(\bm{X}\right)
=&~\P_T\left( \textstyle{\frac{1}{p}}\sum_a \delta_a\l  \P_T\bm{X},\G\G^*\bm{H}_a\r\bm{H}_a- \G\G^*\P_T\bm{X}\right)\cr
=&~\sum_a \left(p^{-1}\delta_a-1\right)\l \bm{X}, \P_T\bm{H}_a\r \P_T\bm{H}_a.
\end{split}
\end{equation*}
Denote $\bm{\Y}_a:\bm{X}\mapsto\left(\textstyle{\frac{\delta_a}{p}}-1\right)\l \bm{X}, \P_T\bm{H}_a\r \P_T\bm{H}_a$. Thus $\mathbb{E}\left[\bm{\Y}_a\right]=0$ and
\begin{align*}
\lV \bm{\Y}_a\rV=&~\mathop{\sup}_{\lV \bm{X}\rV_\fro=1}\lV \left(p^{-1}\delta_a-1\right)\l \bm{X}, \P_T\bm{H}_a\r \P_T\bm{H}_a\rV_\fro \\
\le&~ p^{-1}\mathop{\sup}_{\lV \bm{X}\rV_\fro=1}\lV \bm{X}\rV_\fro \lV \P_T\bm{H}_a\rV_\fro^2
\le 2p^{-1}c_s\mu r/n,
\end{align*}
 where the last inequality follows from
\begin{align*}
\lV\P_T\bm{H}_a\rV_\fro^2\le &~\lV\P_{\bm{U}^{\natural}}\bm{H}_a\rV_\fro^2+\lV\P_{\bm{V}^{\natural}}\bm{H}_a\rV_\fro^2
=\lV\bm{U}^{\natural}\bm{U}^{\natural *}\bm{H}_a\rV_\fro^2+\lV\bm{H}_a\bm{V}^{\natural}\bm{V}^{\natural *}\rV_\fro^2\\
=&~\lV\bm{U}^{\natural *}\bm{H}_a\rV_\fro^2+\lV\bm{H}_a\bm{V}^{\natural}\rV_\fro^2
\le \lV\bm{U}^{\natural}\rV_{2,\infty}^2+\lV\bm{V}^{\natural}\rV_{2,\infty}^2
\le 2c_s\mu r/n.
\end{align*}
 Moreover, we notice that $\bm{\Y}_a$ is self-adjoint
and it holds
\begin{equation*}
\begin{split}
\bm{\Y}_a^{2}(\bm{X})=&~\bm{\Y}_a\left(\left(p^{-1}\delta_a-1\right)\l \bm{X}, \P_T\bm{H}_a\r \P_T\bm{H}_a\right)\\
=&~\left(p^{-1}\delta_a-1\right)^2 \l \l \bm{X}, \P_T\bm{H}_a\r \P_T\bm{H}_a,\P_T\bm{H}_a\r\P_T\bm{H}_a\\
=&~\left(p^{-1}\delta_a-1\right)^2 \lV\P_T\bm{H}_a\rV_\fro^2\l \bm{X}, \P_T\bm{H}_a\r\P_T\bm{H}_a.
\end{split}
\end{equation*}
Therefore, we have
\begin{align*}
\lVert \mathbb{E}(\sum_a \bm{\Y}_a^{2})\lVert
=&\mathop{\sup}_{\lV \bm{X}\rV_\fro=1}\lVert \mathbb{E}\big[\sum_a (p^{-1}\delta_a-1)^2 \lV\P_T\bm{H}_a\rV_\fro^2\l \bm{X}, \P_T\bm{H}_a\r\P_T\bm{H}_a\big]\lVert_\fro\cr
\le&~p^{-1}\textstyle{\mathop{\sup}_{\lV \bm{X}\rV_\fro=1}\lV \left(\sum_a\lV\P_T\bm{H}_a\rV_\fro^2\l \bm{X}, \P_T\bm{H}_a\r\P_T\bm{H}_a\right)\rV_\fro}\cr
\le&~ p^{-1} \mathop{\max}_{a}\textstyle{\lV\P_T\bm{H}_a\rV_\fro^2\mathop{\sup}_{\lV \bm{X}\rV_\fro=1}\lV \left(\sum_a\l \bm{X}, \P_T\bm{H}_a\r\P_T\bm{H}_a\right)\rV_\fro}\cr
= &~p^{-1} \mathop{\max}_{a}\textstyle{\lV\P_T\bm{H}_a\rV_\fro^2\mathop{\sup}_{\lV \bm{X}\rV_\fro=1}\lV\P_T \P_{\Omega}\P_T(\bm{X})\rV_\fro}\cr
\le &~p^{-1}\mathop{\max}_{a}\textstyle{\lV\P_T\bm{H}_a\rV_\fro^2}
~\le 2p^{-1}c_s\mu r/n.
\end{align*}
where the last inequality follows from $\lV\P_T \P_{\Omega}\P_T(\bm{X})\rV_\fro^2=\l \P_T \P_{\Omega}\P_T(\bm{X}),\P_{\Omega}\P_T(\bm{X})\r\le \lV \P_{\Omega}\P_T(\bm{X})\rV_\fro\lV\P_T \P_{\Omega}\P_T(\bm{X})\rV_\fro$ and thus $\lV\P_T \P_{\Omega}\P_T(\bm{X})\rV_\fro\le \lV\P_{\Omega}\P_T(\bm{X})\rV_\fro\le \lV \bm{X}\rV_\fro$. Then, by the Bernstein's inequality \cite[Theorem~1.6]{tropp2012user}, we have
\begin{equation*}
\mathbb{P}\Big(\Big\| \sum_a \bm{\Y}_a\Big\|>t\Big)\le \left(n_1+n_2\right)\exp\left(\frac{-p n t^2 /2}{2c_s \mu r +2c_s\mu r t/3}\right).
\end{equation*}
For any $\varepsilon_0>0$, let $t=\varepsilon_0>0$, and $p\ge c_4 (\varepsilon_0^{-2}c_1 \mu r\log n)/n$ for some universal constant $c_4>12+12\varepsilon_0$. The above inequality then implies $\mathbb{P}(\| \sum_a \bm{\Y}_a\|>\varepsilon_0)\le 2n^{-2}$.
\end{proof}

\begin{lemma}\label{projectionerr1}
For any $\bm{U}\in\mathbb{C}^{n_1\times r}$ and  $\bm{V}\in\mathbb{C}^{n_2\times r}$, if $p\ge \left( \log n\right)/n$, then it holds
\begin{equation}\label{eq:projectionRIP2}
p^{-1}\lV\P_\Omega \left(\bm{UV}^*\right)\rV_\fro^2\le  \lV \bm{UV}^*\rV_\fro^2 +\sqrt{8 p^{-1}n\log n}\lV \bm{U}\rV_\fro\lV \bm{V} \rV_\fro \lV \bm{U}\rV_{2,\infty}\lV \bm{V}\rV_{2,\infty}
\end{equation}
 with probability at least $1-2n^{-2}$.
\end{lemma}
\begin{proof}
By the supporting \cref{uv} in \cref{supportlemmas}, we have
\begin{align*}
p^{-1}\lV\P_\Omega \left(\bm{UV}^*\right)\rV_\fro^2=&~p^{-1}\sum_a\bigg| \sum_{i+j=a+1}\frac{\delta_a}{\sqrt{\varsigma_a}}\l \bm{UV}^*,\bm{e}_i \bm{e}_j^\top\r \bigg|^2\\
 \le&~p^{-1}\sum_a\sum_{i+j=a+1}\delta_a\lV \bm{U}_{(i,:)}\rV_2^2 \lV \bm{V}_{(j,:)}\rV_2^2\\
 \le&~ \lV \bm{UV}^*\rV_\fro^2+\sqrt{8p^{-1}n\log n}\sqrt{\sum_i\lV\bm{U}_{(i,:)}\rV_2^4 }\sqrt{\sum_J\lV\bm{V}_{(j,:)}\rV_2^4}\\
 \le&~ \lV \bm{UV}^*\rV_\fro^2+\sqrt{p^{-1}8n\log n}\lV \bm{U}\rV_\fro \lV \bm{U}\rV_{2,\infty}\lV \bm{V} \rV_\fro\lV \bm{V}\rV_{2,\infty}.
\end{align*}
\end{proof}

\begin{lemma}\label{RIP2}
 For any matrix $\bm{A}\in T$,  under event \eqref{eq:projectionRIP1}, it holds
\begin{equation*}
p\left(1-\varepsilon_0\right)\lV \bm{A}\rV_\fro^2\le \lV\P_\Omega\G\G^* \bm{A}\rV_\fro^2 \le p\left(1+\varepsilon_0\right)\lV \bm{A}\rV_\fro^2
\end{equation*}
\end{lemma}
\begin{proof}
For $\bm{A}\in T$, we have $\P_T \bm{A}=\bm{A}$. Thus, by \cref{RIP1} we have
\begin{equation*}
\begin{aligned}
\lV \P_{\Omega} \G\G^*\bm{A}\rV_\fro^2=\l \P_T\P_{\Omega}\G\G^* \P_T \bm{A},\bm{A}\r
\le \lV \P_T\P_{\Omega}\G\G^* \P_T \bm{A}\rV_\fro  \lV\bm{A}\rV_\fro \le p(1+\varepsilon_0) \lV\bm{A}\rV_\fro^2.
\end{aligned}
\end{equation*}
Similarly, $p\left(1-\varepsilon_0\right)\lV \bm{A}\rV_\fro^2\le \lV\P_\Omega\G\G^* \bm{A}\rV_\fro^2$. 
\end{proof}

\subsection{Proof of \texorpdfstring{\cref{thm:Initialization}}{Theorem~\protect\ref{thm:Initialization}} (guaranteed initialization)}\label{subsec:proofini}

\begin{proof}[Proof of \cref{thm:Initialization}]
We will finish the proof in three steps under event \cref{event:RIP3}. For simplicity, following the notation in \cref{algo:HSGD}, we denote
\begin{equation}\label{def-M}
 \bm{M}^{(0)}:=p^{-1}\G\left(\Pi_{\Omega}\bm{f}-\bm{s}^{(0)}\right),\quad
 \text{and}~\bm{M}^{(0)}_r~\text{ is the top-$r$ SVD of}~ \bm{M}^{(0)}.
\end{equation}
\textbf{Step 1.} We first bound $\lV \bm{M}^{(0)}-\G \bm{z}^{\natural}\rV_2$. By the triangle inequality we have
\begin{equation}\label{term-tri}
\lV \bm{M}^{(0)}-\G \bm{z}^{\natural}\rV_2\le \lV \bm{M}^{(0)}-p^{-1}\G \Pi_\Omega \bm{z}^{\natural}\rV_2 +\lV p^{-1}\G \Pi_\Omega \bm{z}^{\natural}-\G \bm{z}^{\natural}\rV_2.
\end{equation}
The definition in \eqref{def-M} yields
\begin{equation*}
\lV \bm{M}^{(0)}-{\frac{1}{p}}\G \Pi_\Omega \bm{z}^{\natural}\rV_2
=\lV {\frac{1}{p}}\G\left(\Pi_{\Omega}\bm{f}-\bm{s}^{(0)}\right)-{\frac{1}{p}}\G \Pi_\Omega \bm{z}^{\natural}\rV_2
=\lV {\frac{1}{p}}\G\left(\Pi_{\Omega}\bm{s}^{\natural}-\bm{s}^{(0)}\right)\rV_2.
\end{equation*}
Denote the support of $\Pi_{\Omega}\bm{s}^{\natural}$ and $\bm{s}^{(0)}$ by $\Omega_{s}^{\natural}$ and $\Omega_s^{(0)}$ respectively. Notice that $\Omega_{s}^{\natural}\subseteq \Omega$, $\Omega_{s}^{(0)}\subseteq \Omega $. By the definition of $\bm{s}^{(0)}$ in \cref{algo:HSGD}, we have
\begin{equation*}
\begin{aligned}
\left[\D^{-1}(\Pi_{\Omega}\bm{s}^{\natural}-\bm{s}^{(0)})\right]_i&=-(\D^{-1}\bz^{\natural})_i, \quad \text{for}~ i\in \Omega_{s}^{\natural}\cap \Omega_s^{(0)}~ \text{and}~ i\in  \Omega_s^{(0)} \backslash \Omega_{s}^{\natural}.
\end{aligned}
\end{equation*}
Recall $\Pi_{\Omega}\bm{f}=\Pi_{\Omega} \bm{z}^{\natural}+\Pi_{\Omega}\bm{s}^{\natural}$ and $\lV \Pi_{\Omega}\bm{s}^{\natural}\rV_0\le\alpha pn$. We see there are no more than $\alpha pn$ elements in $\D^{-1}\Pi_{\Omega}\bm{f}$ such that $|(\D^{-1}\bf)_i|> \lV \D^{-1}\Pi_{\Omega}\bm{z}^{\natural}\rV_{\infty}$. Also, by the definition of the operator $\Gamma_{\alpha p}$, we know $|(\D^{-1}\bf)_i|\le \lV \D^{-1}\Pi_{\Omega}\bm{z}^{\natural}\rV_{\infty}$ for all $i\in \Omega_{s}^{\natural}\setminus\Omega_s^{(0)}$. Using these facts, we obtain
$$\left[\D^{-1}(\Pi_{\Omega}\bm{s}^{\natural}-\bm{s}^{(0)})\right]_i=(\D^{-1}s^{\natural})_i=\left(\D^{-1}\bm{f}-\D^{-1}\bm{z}^{\natural}\right)_i\le |(\D^{-1}\bf)_i|+|(\D^{-1}\bz^{\natural})_i|\le 2\lV \D^{-1}\Pi_{\Omega}\bm{z}^{\natural}\rV_{\infty}$$
for$\ i\in \Omega_{s}^{\natural}\setminus\Omega_s^{(0)}$. Combing all the pieces for $\D^{-1}\left(\Pi_{\Omega}\bm{s}^{\natural}-\bm{s}^{(0)}\right)$, we obtain
\begin{equation}\label{s-s}
 \big\|\D^{-1}(\Pi_{\Omega}\bm{s}^{\natural}-\bm{s}^{(0)})\big\|_{\infty}\le 2\lV \D^{-1}\Pi_{\Omega}\bm{z}^{\natural}\rV_{\infty}.
 \end{equation}
Since $\G=\H\D^{-1}$, by the definition of $\H$ and $\D$ we know $\lV\G \left(\Pi_{\Omega}\bm{s}^{\natural}-\bm{s}^{(0)}\right)\rV_{\infty}\le 2\lV \G \Pi_{\Omega}\bm{z}^{\natural}\rV_{\infty}$. It then yields
\begin{equation} \label{term1}
\big\|\G \big(\Pi_{\Omega}\bm{s}^{\natural}-\bm{s}^{(0)}\big)\big\|_{2}\le 2\alpha p n\big\| \G\big(\Pi_{\Omega}\bm{s}^{\natural}-\bm{s}^{(0)}\big)\big\|_{\infty}
                                           \le  4\alpha p n\big\| \G \bm{z}^{\natural}\big\|_{\infty}
                                           \le  4\alpha p c_s\mu r \sigma_{1}^{\natural},
\end{equation}
where the first inequality follows from $\lV\Pi_{\Omega}\bm{s}^{\natural}-\bm{s}^{(0)}\rV_{0}\le 2\alpha pn$ and \cref{s-matrix}. And the last inequality follows from
$\big\| \G \bm{z}^{\natural}\big\|_{\infty}\le \big\| \bm{L}^{\natural} \big\|_{2,\infty} \big\|\G \bm{z}^{\natural}\big\|_{2} \big\| \bm{R}^{\natural}\big\|_{2,\infty}\le \mu \frac{c_sr}{n}\lV \G \bm{z}^{\natural}\rV_{2}$. Finally, under event \eqref{event:RIP3}, we combine \eqref{term-tri} and \eqref{term1} to obtain
\begin{equation}\label{term3}
\lV \bm{M}^{(0)}-\G \bm{z}^{\natural}\rV_2
\le 4\alpha c_s\mu r \sigma_{1} ^{\natural} +c_0\sqrt{(p^{-1}\mu c_s r\log n)/n}\sigma_{1} ^{\natural}
\le 4\alpha c_s\mu r\kappa \sigma_{r} ^{\natural}+\frac{c_0\varepsilon_0\sigma_r^{\natural}}{\sqrt{\kappa r}},
\end{equation}
provided $p\ge (\varepsilon_0^{-2}\kappa^3 c_s\mu r^2\log n)/n$

\textbf{Step 2.} The second step is to bound $\tilde{d}_0$. First we have
\begin{align}\label{ineq:M0rGz}
\lV \bm{M}^{(0)}_r-\G \bm{z}^{\natural}\rV_2 &\le \lV \bm{M}^{(0)}_r-\bm{M}^{(0)}\rV_2 +\lV \bm{M}^{(0)}-\G \bm{z}^{\natural}\rV_2 \cr
 &\le 2\lV \bm{M}^{(0)}-\G \bm{z}^{\natural}\rV_2
 ~\le 8\alpha c_s\mu r\kappa \sigma_{r} ^{\natural}+2c_0\varepsilon_0\sigma_r^{\natural}(\kappa r)^{-1/2},
 \end{align}
where the second inequality follows from the definition of $\bm{M}^{(0)}_r$ and Eckart-Young-Mirsky theorem, the third inequality follows from \eqref{term3}. Hence, for any $\varepsilon_0\in \left(0,\frac{\sqrt{\kappa r}}{8c_0}\right)$, as long as $\alpha\le \frac{1}{32 c_s\mu r\kappa}$ one has
\begin{equation}\label{U0V0-M}
\lV \bm{M}^{(0)}_r-\G \bm{z}^{\natural}\rV_2\le {\frac{1}{2}}\sigma_r^{\natural}.
\end{equation}
Then, by the definition of $\tilde{d}_0$ and the inequality in \cite[Lemma~5.14]{tu2016low}, we have
\begin{equation}\label{dbound}
\tilde{d}_0^2\le {\frac{2}{(\sqrt{2}-1)\sigma_r^{\natural}}\lV \bm{M}^{(0)}_r-\G \bm{z}^{\natural}\rV_\fro^2}
~\le {\frac{10r}{\sigma_r^{\natural}}\lV \bm{M}^{(0)}_r-\G \bm{z}^{\natural}\rV_2^2},
\end{equation}
which together with \cref{ineq:M0rGz} reveals
$$\tilde{d}_0\le 26\alpha c_s\kappa \mu r\sqrt{r} \sqrt{\sigma_{r} ^{\natural}}+7c_0\varepsilon_0\sqrt{\sigma_{r} ^{\natural}/\kappa}.$$

\textbf{Step 3.} Now we present the final step, which is to show $d_0 \le \tilde{d}_0$. By \eqref{U0V0-M} and Weyl's theorem \cite{bhatia2013matrix}, we obtain
\begin{equation}\label{ineq:L0sigma}
\sqrt{\sigma_{1} ^{\natural}/2}\le \big\| \tilde{\bm{L}}^{(0)}\big\|_2\le \sqrt{3\sigma_{1} ^{\natural}/2},\  \sqrt{\sigma_{1} ^{\natural}/2}\le \big\| \tilde{\bm{R}}^{(0)}\big\|_2\le \sqrt{3\sigma_{1} ^{\natural}/2}.
\end{equation}
Together with the $\mu$-incoherence of $\G \bm{z}^{\natural}$, it yields $\left(\bm{L}^{\natural}, \bm{R}^{\natural}\right) \in \L \times \R$. Let $\tilde{\bm{Q}}$ be the best align matrix between     $\left(\tilde{\bm{L}}^{(0)},\tilde{\bm{R}}^{(0)}\right)$ and $\left(\bm{L}^{\natural},\bm{R}^{\natural} \right)$. We then have
\begin{align}\label{non-expansion}
 \mathrm{dist}^2\left(\bm{L}^{(0)},\bm{R}^{(0)};\bm{L}^{\natural},\bm{R}^{\natural}\right)
\le&~\lV \Pi_{\L}\left(\tilde{\bm{L}}^{(0)}\right)-\bm{L}^{\natural} \tilde{\bm{Q}} \rV_\fro^2+\lV \Pi_{\L}\left(\tilde{\bm{R}}^{(0)}\right)-\bm{Z_{V} }^{\natural} \tilde{\bm{Q}} \rV_\fro^2 \cr
\le&~ \lV \tilde{\bm{L}}^{(0)}-\bm{L}^{\natural} \tilde{\bm{Q}} \rV_\fro^2+\lV \tilde{\bm{R}}^{(0)}-\bm{R}^{\natural} \tilde{\bm{Q}} \rV_\fro^2\cr
 =&~\mathrm{dist}^2\left(\tilde{\bm{L}}^{(0)},\tilde{\bm{R}}^{(0)};\bm{L}^{\natural},\bm{R}^{\natural} \right),
\end{align}
where the second inequality comes from the non-expansion property of projections onto the convex sets $\L$ and $\R$. Finally, \eqref{non-expansion} implies $d_0 \le \tilde{d}_0$.
\end{proof}

\subsection{Proof of \texorpdfstring{\cref{thm:convergence}}{Theorem~\protect\ref{thm:convergence}} (local convergence)}\label{subsec:prooflocalcon}
Firstly, we present some key lemmas that is essential for the proof of local convergence. Some parts of the proofs in this section follow similar techniques introduced in \cite{yi2016fast,cai2018spectral}. We begin with the following definition.
\begin{definition}\label{def1}
Let $\left(\bm{L},\bm{R}\right)$ be arbitrary matrices in the the space $\left(\L\times \R\right) \cap \B(\sqrt{\sigma_1^{\natural}})$.\footnote{$\B(\cdot)$ is the ball with the centre $\left(\bm{L}^\natural,\bm{R}^\natural\right)$ and a radius defined by the distance in \eqref{def-error}.} Define $\bm{s}$ as
\begin{equation*}
\bm{s}:=\Gamma_{\gamma\alpha p}\left(\Pi_{\Omega}\left(\bm{f}-\G^*\left(\bm{L} \bm{R}^*\right)\right)\right).
\end{equation*}
Define $\left(\bm{L}_{\E},\bm{R}_{\E}\right)$ as the solution set which satisfies
\begin{equation*}
\left(\bm{L}_{\E},\bm{R}_{\E}\right)\in \argmin_{(\widehat{\bm{L}},\widehat{\bm{R}})\in\E\left(\bm{L}^{\natural},\bm{R}^{\natural}\right)} \big\| \widehat{\bm{L}}-\bm{L} \big\|_\fro^2 +\big\| \widehat{\bm{R}}-\bm{R} \big\|_\fro^2.
\end{equation*}
Define $\BDeltaL$, $\BDeltaR$ and $\Delta$ as
\begin{equation*}
\BDeltaL:=\bm{L}-\bm{L}_{\E},\quad \BDeltaR:=\bm{R}-\bm{R}_{\E},\quad \textnormal{and } \Delta:=\big\| \BDeltaL \big\|_\fro^2+ \big\| \BDeltaR \big\|_\fro^2.
\end{equation*}
\end{definition}
Actually, in \cref{def1}, $\left(\bm{L}_{\E},\bm{R}_{\E}\right)$ is aligned with some $\bm{Q}\in \mathbb{Q}_r$ to be the solution that match $\left(\bm{L},\bm{R}\right)$ best. So the error between $\left(\bm{L},\bm{R}\right)$ and the solution set is then defined by $\BDeltaL$, $\BDeltaR$ and $\Delta$. Let the iteration sequence $\left\{\bm{L}^{(k)},\bm{R}^{(k)}\right\}_{k\ge 1}$ be generated by the gradient descent strategy described in \eqref{eq:update_Hankel}. 
Denote $\nabla \ell^{(k)}:=\nabla\ell\left(\bm{L}^{(k)},\bm{R}^{(k)};\bm{s}^{(k+1)}\right)$. In fact, if there is a proper lower bound for the term
$\mathrm{Re}(\langle \nabla_{\bm{L}} \ell^{(k)},\BDeltaL^{(k)}\rangle+\langle  \nabla_{\bm{R}} \ell^{(k)},\BDeltaR^{(k)}\rangle)$
and also a proper upper bound for the term $\|\nabla_{\bm{L}} \ell^{(k)} \|_\fro^2+\| \nabla_{\bm{R}} \ell^{(k)} \|_\fro^2$, we can then show the local convergence of HSGD (see \eqref{error:dk} for details).

In the rest of this section, we give bounds for the above terms in \cref{T1low,T2up}, and the proof of \cref{thm:convergence} follows.
Some crucial lemmas have to be shown first for the bounds.


\begin{lemma}\label{PZ-Z}
Let $\BDeltaL,\BDeltaR, \bm{L}, \bm{R},\bm{L}_{\E},\bm{R}_{\E}$ and $\Delta$ be defined in \cref{def1}.  Then, if provided $p\ge (\varepsilon_0^{-2}c_4 c_s^2\mu^2 r^2 \log n)/n$, $\alpha\le \frac{1}{32 c_s\mu r\kappa}$, $\varepsilon_0\in \left(0,\frac{\sqrt{\kappa r}}{8c_0}\right)$, under events \cref{event:RIP3}, \eqref{eq:projectionRIP1} and \eqref{eq:projectionRIP2} it holds
\begin{align}
 \lV\P_{\Omega}\left(\BDeltaL \BDeltaR^*\right)\rV_\fro^2 \le&~ {\frac{p}{4}}\Delta^2+18\varepsilon_0 p\sigma_1^{\natural}\Delta, \cr
 \lV\P_{\Omega}\G\G^*\left(\bm{L} \bm{R}^*-\bm{L}_{\E}\bm{R}_{\E}^*\right)\rV_\fro^2 \le&~ 4\left(1+10\varepsilon_0\right)p\sigma_1^{\natural}\Delta+{\frac{p}{2}}\Delta^2.
\end{align}
\end{lemma}
\begin{proof}
By \cref{projectionerr1}, we have
\begin{equation*}
\begin{aligned}
&~\lV\P_{\Omega}\left( \BDeltaL \BDeltaR^*\right)\rV_\fro^2 \cr
\le&~ p\lV \BDeltaL \rV_\fro^2\lV \BDeltaR \rV_\fro^2+\sqrt{8p n\log n}\lV \BDeltaL\rV_\fro\lV \BDeltaR \rV_\fro \lV \BDeltaL\rV_{2,\infty}\lV \BDeltaR\rV_{2,\infty}\cr
\le &~p\lV \BDeltaL \rV_\fro^2\lV \BDeltaR \rV_\fro^2+36\varepsilon_0 p\sigma_1^{\natural}\lV \BDeltaL\rV_\fro\lV \BDeltaR \rV_\fro
\le  {\frac{p}{4}}\Delta^2+18\varepsilon_0 p\sigma_1^{\natural}\Delta,
\end{aligned}
\end{equation*}
where the second inequality follows from $p\ge (\varepsilon_0^{-2}c_s^2\mu^2 r^2 \log n)/n$ and (by \cref{ineq:L0sigma})
\begin{equation}\label{HU-2norm}
\begin{aligned}
\lV \BDeltaL \rV_{2,\infty}\le \lV \bm{L} \rV_{2,\infty}+\lV \bm{L}_{\E} \rV_{2,\infty}\le 2\sqrt{2\mu r c_s n^{-1}} \lV \bm{\tilde{L}^{(0)}}\rV_2 \le 2\sqrt{3\mu r c_s\sigma^{\natural}_1 n^{-1}}.
\end{aligned}
\end{equation}
Then, we have
\begin{align}\label{ineq:projLR}
&~\lV\P_{\Omega}\G\G^*\left( \bm{L} \bm{R}^*-\bm{L}_{\E}\bm{R}_{\E}^*\right)\rV_\fro^2 \cr
\le&~ 2\lV\P_{\Omega}\G\G^*\left( \bm{L}_{\E} \BDeltaR^*+\BDeltaL \bm{R}_{\E}^*\right)\rV_\fro^2+2\lV\P_{\Omega}\G\G^*\left(\BDeltaL \BDeltaR^*\right)\rV_\fro^2\cr
\le &~ 2\left(1+\varepsilon_0\right)p\lV \bm{L}_{\E} \BDeltaR^*+\BDeltaL \bm{R}_{\E}^*\rV_\fro^2+{\frac{p}{2}}\Delta^2+36\varepsilon_0 p\sigma_1^{\natural}\Delta\cr
\le&~ 4\left(1+\varepsilon_0\right)p\left(\lV \bm{L}_{\E} \BDeltaR^* \rV_\fro^2+\lV \BDeltaL \bm{R}_{\E}^*\rV_\fro^2\right)
+{\frac{p}{2}}\Delta^2+36\varepsilon_0 p\sigma_1^{\natural}\Delta\cr
\le &~ 4\left(1+10\varepsilon_0\right)p\sigma_1^{\natural}\Delta+{\frac{p}{2}}\Delta^2,
\end{align}
where the second inequality follows from \cref{RIP2}.
\end{proof}

\begin{lemma}\label{PZUZV-M}
Let $\bm{L}, \bm{R}, \bm{L}_{\E},\bm{R}_{\E}$ and $\Delta$ be defined in \cref{def1}. Then, under same condition to \cref{PZ-Z}, if $\lv\Omega \lv \le \alpha n$, we have
$$\lV \P_{\Omega}\left(\bm{L} \bm{R}^*-\bm{L}_{\E} \bm{R}_{\E}^*\right) \rV_\fro^2\le 18\alpha c_s\mu r\sigma_1^{\natural}\Delta.$$
\end{lemma}
\begin{proof}
By \cref{def1} we have
\begin{equation*}
\begin{aligned}
&~\lv\left[\bm{L} \bm{R}^*-\bm{L}_{\E} \bm{R}_{\E}^*\right]_{i,j}\lv
=\lv\left[\bm{L}_{\E} \left(\bm{R}^*-\bm{R}_{\E}^*\right)+\left(\bm{L}-\bm{L}_{\E}\right) \bm{R}_{\E}^*+\left(\bm{L}-\bm{L}_{\E}\right) \left(\bm{R}^*-\bm{R}_{\E}^*\right)\right]_{i,j}\lv \\
\le&~ \lV \bm{L}_{\E}\rV_{2,\infty} \lV \BDeltaR\left(j,:\right)\rV_2
+\lV \bm{R}_{\E}\rV_{2,\infty} \lV \BDeltaL\left(i,:\right)\rV_2\\
&~+{\frac{1}{2}}\left(\lV \BDeltaL\rV_{2,\infty} \lV  \BDeltaR\left(j,:\right)\rV_2+\lV \BDeltaR\rV_{2,\infty} \lV  \BDeltaL\left(i,:\right)\rV_2 \right)\\
\le&~ \left(1+\sqrt{3}\right)\sqrt{c_s\mu r \sigma_1^{\natural}n^{-1}}\lV \BDeltaR\left(j,:\right)\rV_2+\left(1+\sqrt{3}\right)\sqrt{c_s\mu r \sigma_1^{\natural}n^{-1}}\lV \BDeltaL\left(i,:\right)\rV_2,
\end{aligned}
\end{equation*}
where the last inequality follows from \eqref{HU-2norm}. Denote $\Phi:=\{\left(i,j\right):\left(i+j-1\right)\in\Omega\}$. The Hankel structure yields $\mathrm{card}\left(\Phi\left(i,:\right)\right)\le\alpha n$, $\mathrm{card}\left( \Phi\left(:,j\right)\right)\le\alpha n$. Hence, we use the inequality $(a+b)^2\le 2a^2+2b^2$ and the fact $\sum_{\left(i,j\right)\in \Phi}\lV \BDeltaR\left(j,:\right)\rV_2=\sum_{i\in \Phi\left(:,j\right)}\sum_j\lV \BDeltaR\left(j,:\right)\rV_2$ to obtain
\begin{align*}
&~\lV \P_{\Omega}\left(\bm{L} \bm{R}^*-\bm{L}_{\E} \bm{R}_{\E}^*\right) \rV_\fro^2 \le \sum_{\left(i,j\right)\in \Phi}\lv [\bm{L} \bm{R}^*-\bm{L}_{\E} \bm{R}_{\E}^*]^2_{i,j}\lv\\
    \le&~\sum_{\left(i,j\right)\in \Phi}18c_s\mu r \sigma_1^{\natural}n^{-1}\left(\lV \BDeltaR\left(j,:\right)\rV_2^2+\lV \BDeltaL\left(i,:\right)\rV_2^2\right)~
    \le 18\alpha c_s\mu r \sigma_1^{\natural}\Delta.
\end{align*}
\end{proof}

With \cref{RIP2,PZ-Z,PZUZV-M} in hand, we can give the bounds in the following \cref{T1low,T2up} that lead to the local descent property of $\ell$ .
\begin{lemma}\label{T1low} Let $\BDeltaL,\BDeltaR, \bm{L}, \bm{R},\bm{L}_{\E},\bm{R}_{\E}$, and $\Delta$ be defined in \cref{def1}. Set $\lambda=\frac{1}{16}$. Let $\gamma\in \left[1+\frac{1}{b_0},2\right]$ with any given $1\le b_0<\infty$. If provided $p\ge \left(\varepsilon_0^{-2}c_4 c_s^2\mu^2 r^2 \log n\right)/n$, $\alpha\le \frac{1}{32 c_s\mu r\kappa}$, then for any $\varepsilon_0\in(0,\frac{1}{10})\cap \left(0,\frac{\sqrt{\kappa r}}{8c_0}\right)$, under events \cref{event:RIP3}, \eqref{eq:projectionRIP1} and \eqref{eq:projectionRIP2} we have
\begin{equation*}
\begin{aligned}
&~\mathrm{Re}\left(\l \nabla_{\bm{L}} \ell,\bm{L}-\bm{L}_{\E}\r+\l  \nabla_{\bm{R}} \ell,\bm{R}-\bm{R}_{\E}\r\right) \\
\ge&~{\frac{7}{8}}\lV \bm{L} \bm{R}^*
-\bm{L}_{\E} \bm{R}_{\E}^*\rV_\fro^2+(\frac{1}{64}\sigma_r^\natural-\nu\sigma_1^{\natural})\Delta-\left(2+ \frac{\sqrt{2b_0}}{4}+\frac{b_0}{4\beta}\right)\Delta^2\\
 &~-\left(\sqrt{3b_0}+6\sqrt{c_s\alpha\mu r}+{\frac{\sqrt{2}}{16}}\right)\sqrt{\sigma_1^{\natural}\Delta^3}
 +{\frac{1}{64}}\lV\bm{L}^*\bm{L}-\bm{R}^*\bm{R}\rV_\fro^2,
\end{aligned}
\end{equation*}
where $\nu:=\left(54 +9\beta\right)\alpha\mu r+4b_0\beta^{-1}+\varepsilon_1$ with $\varepsilon_1:=\sqrt{\varepsilon_0}\left(12\sqrt{b_0}+44\sqrt{c_s\alpha \mu r}\right)+41\varepsilon_0$ and  arbitrary $\beta>0$.
\end{lemma}
\begin{proof}
The proof are divided into two main parts. In the first part of the proof, we establish a lower bound of $\mathrm{Re}\l \nabla_{\bm{L}} \psi,\bm{L}-\bm{L}_{\E} \r+\mathrm{Re}\l \nabla_{\bm{R}} \psi,\bm{R}-\bm{R}_{\E}\r$. In fact,
\begin{align*}
&\mathrm{Re}\l \nabla_{\bm{L}} \psi,\bm{L}-\bm{L}_{\E} \r
=\mathrm{Re}\l p^{-1}\G \Pi_{\Omega}\left(\G^*\left(\bm{L} \bm{R}^*\right)+\bm{s}-\bm{f}\right)+\left(\I-\G\G^*\right)\left(\bm{L} \bm{R}^*\right),\BDeltaL \bm{R}^*\r,\\
&\mathrm{Re}\l \nabla_{\bm{R}} \psi,\bm{R}-\bm{R}_{\E}\r
=\mathrm{Re}\l p^{-1}\G \Pi_{\Omega}\left(\G^*\left(\bm{L} \bm{R}^*\right)+\bm{s}-\bm{f}\right)+\left(\I-\G\G^*\right)\left(\bm{L} \bm{R}^*\right),\bm{L} \BDeltaR^*\r.
\end{align*}
Notice $\left(\I-\G\G^*\right)\left(\bm{L}_{\E} \bm{R}_{\E}^*\right)=0$ and  $\bm{z}^{\natural}=\G^*\left(\bm{L}_{\E} \bm{R}_{\E}^*\right)$. Rearrangement gives
\begin{equation*}
\begin{split}
&~\mathrm{Re}\left(\l \nabla_{\bm{L}} \psi,\bm{L}-\bm{L}_{\E}\r+\l  \nabla_{\bm{R}} \psi,\bm{R}-\bm{R}_{\E}\r\right)\\
=&~\underbrace{\mathrm{Re}\l \G\left(p^{-1}\Pi_{\Omega}-\I\right)\G^*\left(\bm{L} \bm{R}^*-\bm{L}_{\E} \bm{R}_{\E}^*\right),\BDeltaL \bm{R}^*+\bm{L} \BDeltaR^*\r}_{T_1}\\
&~+\underbrace{\mathrm{Re}\l \bm{L} \bm{R}^*-\bm{L}_{\E} \bm{R}_{\E}^*,\BDeltaL \bm{R}^*+\bm{L} \BDeltaR^*\r}_{T_2}
+\underbrace{\mathrm{Re}\l p^{-1}\G \Pi_{\Omega}\left(\bm{s}-\bm{s}^{\natural}\right),\BDeltaL \bm{R}^*+\bm{L} \BDeltaR^*\r}_{T_3}.
\end{split}
\end{equation*}
Now we estimate the bounds of $T_1, T_2, T_3$.
By direct calculation, we have
\begin{equation}\label{Hu}
\BDeltaL \bm{R}^*+\bm{L} \BDeltaR^*
=\bm{L} \bm{R}^*-\bm{L}_{\E} \bm{R}_{\E}^*+\BDeltaL \BDeltaR^*
=\BDeltaL \bm{R}_{\E}^*+\bm{L}_{\E}\BDeltaR^*+2\BDeltaL \BDeltaR^*.
\end{equation}
Then, for $T_1$ we have
\begin{equation*}
\begin{aligned}
 | T_1 |
 =&~\Big\lv\mathrm{Re}\Big\langle \G\left(p^{-1}\Pi_{\Omega}-\I\right)\G^*\left(\BDeltaL \bm{R}_{\E}^*+\bm{L}_{\E}\BDeltaR^*+\BDeltaL \BDeltaR^*\right),\\
                    &~\BDeltaL \bm{R}_{\E}^*+\bm{L}_{\E}\BDeltaR^*+2\BDeltaL \BDeltaR^*\Big\rangle \Big\lv\\
    \le &~ \left\lv \mathrm{Re}\l \G\left(p^{-1}\Pi_{\Omega}-\I\right)\G^*\left(\BDeltaL \bm{R}_{\E}^*+\bm{L}_{\E}\BDeltaR^*\right),\left(\BDeltaL \bm{R}_{\E}^*+\bm{L}_{\E}\BDeltaR^*\right)\r\right\lv\\
    &~+2\left\lv \mathrm{Re}\l\G\left(p^{-1}\Pi_{\Omega}-\I\right)\G^*\left(\BDeltaL \BDeltaR^*\right),\BDeltaL \BDeltaR^*\r\right\lv\\
    &~  +3\lv\l \G\left(p^{-1}\Pi_{\Omega}-\I\right)\G^*\left(\BDeltaL \bm{R}_{\E}^*+\bm{L}_{\E}\BDeltaR^*\right),\BDeltaL \BDeltaR^*\r\lv\\
    \end{aligned}
    \end{equation*}
Notice that  $\P_T\left(\BDeltaL \bm{R}_{\E}^*+\bm{L}_{\E}\BDeltaR^*\right)=\BDeltaL \bm{R}_{\E}^*+\bm{L}_{\E}\BDeltaR^*$, thus
\begin{align}\label{T1}
 | T_1 |  \le &~ \varepsilon_0\lV  \BDeltaL\bm{R}_{\E}^*+\bm{L}_{\E}\BDeltaR^*\rV_\fro^2+2p^{-1}\lV\P_{\Omega}\G\G^*\BDeltaL \BDeltaR^*\rV_\fro^2-2\lV \G\G^*\BDeltaL \BDeltaR^*\rV_\fro^2 \cr
     &~+{\frac{3 \varepsilon_0}{2}}\lV\BDeltaL\bm{R}_{\E}^*+\bm{L}_{\E}\BDeltaR^*\rV_\fro^2+\frac{3}{2}\lV\BDeltaL \BDeltaR^*\rV_\fro^2 \cr
    \le&~ {\frac{5\varepsilon_0}{2}}\lV\BDeltaL\bm{R}_{\E}^*+\bm{L}_{\E}\BDeltaR^*\rV_\fro^2
    +\frac{1}{2}\Delta^2+36\varepsilon_0\sigma_1^{\natural}\Delta+\frac{3}{2}\lV\BDeltaL \BDeltaR^*\rV_\fro^2 \cr
    \le&~ 41\varepsilon_0\sigma_1^{\natural}\Delta+{\frac{7}{8}}\Delta^2,
\end{align}
where the first inequality follows from \cref{RIP2} and $ab\le \frac{a^2+b^2}{2}$, the second inequality follows from $\lV\P_{\Omega}\G\G^*\bm{X}\rV_\fro=\lV\G\G^*\P_{\Omega} \BX \rV_\fro$ for $\BX\in \mathbb{C}^{n_1\times n_2}$, then apply \cref{PZ-Z}, and the fact $\G\G^*$ is a projection operator. The last inequality follows from $\lV\BDeltaL\bm{R}_{\E}^*\rV_\fro^2 \le \sigma_1^{\natural}\lV \BDeltaL\rV_\fro^2$,  $\lV \bm{L}_{\E}\BDeltaR^*\rV_\fro^2\le \sigma_1^{\natural}\lV \BDeltaR\rV_\fro^2$ and the Cauchy-Schwarz inequality which yields $\lV\BDeltaL \BDeltaR^*\rV_\fro\le \lV\BDeltaL\rV_\fro \lV \BDeltaR\rV_\fro \le\frac{\Delta}{2} $.

About $T_2$, by $\BDeltaL \bm{R}^*+\bm{L} \BDeltaR^*=\bm{L} \bm{R}^*-\bm{L}_{\E} \bm{R}_{\E}^*+\BDeltaL \BDeltaR^*$ we have
\begin{align}\label{T2}
T_2&\ge \lV \bm{L} \bm{R}^*-\bm{L}_{\E} \bm{R}_{\E}^*\rV_\fro^2-\lV \bm{L} \bm{R}^*-\bm{L}_{\E} \bm{R}_{\E}^*\rV_\fro\lV \BDeltaL \BDeltaR^*\rV_\fro \cr
&\ge {\frac{15}{16}}\lV \bm{L} \bm{R}^*-\bm{L}_{\E} \bm{R}_{\E}^*\rV_\fro^2-4\lV \BDeltaL \BDeltaR^*\rV_\fro^2 \ge \frac{15}{16}\lV \bm{L} \bm{R}^*-\bm{L}_{\E} \bm{R}_{\E}^*\rV_\fro^2-\Delta^2,
\end{align}
where the second inequality follows from $a^2-ab\ge \frac{15}{16}a^2-4b^2$.

About $T_3$, we denote the support of $\bm{s}$ and $\bm{s}^{\natural}$ by $\Omega_{s}$ and $\Omega_{\natural}$, respectively. We have
\begin{align*}
p\lv T_3\lv &\le \lv \mathrm{Re}\l\G \Pi_{\Omega}\left(\bm{s}-\bm{s}^{\natural}\right),\bm{L} \bm{R}^*-\bm{L}_{\E} \bm{R}_{\E}^*\r\lv +\lv \mathrm{Re}\l \G \Pi_{\Omega}\left(\bm{s}-\bm{s}^{\natural}\right),\BDeltaL \BDeltaR^*\r\lv\\
&\le  \underbrace{\lv\mathrm{Re}\l \P_{\Omega} \G \Pi_{\Omega_s}\left(\bm{s}-\bm{s}^{\natural}\right),\bm{L} \bm{R}^*-\bm{L}_{\E} \bm{R}_{\E}^*\r\lv}_{T_4}\\
&+ \underbrace{\lv\mathrm{Re}\l\P_{\Omega}\G \Pi_{\Omega_{\natural}\setminus \Omega_s}\left(\bm{s}-\bm{s}^{\natural}\right),\bm{L} \bm{R}^*-\bm{L}_{\E} \bm{R}_{\E}^*\r\lv}_{T_5}
+ \underbrace{\lv \mathrm{Re}\l \G \Pi_{\Omega}\left(\bm{s}-\bm{s}^{\natural}\right),\BDeltaL \BDeltaR^*\r\lv}_{T_6},
\end{align*}
Noticing that for $i\in \Omega\cap\Omega_s$, we have $s_i=z^{\natural}_i+s^{\natural}_i-[\G^*\left(\bm{L} \bm{R}^*\right)]_i$, then
\begin{equation*}
\begin{aligned}
 T_4=&~ \lv \mathrm{Re}\l \P_{\Omega}\P_{\Omega_s}[\G \bm{z}^{\natural}-\G\G^*\left(\bm{L} \bm{R}^*\right)],\bm{L} \bm{R}^*-\bm{L}_{\E} \bm{R}_{\E}^*\r\lv \\
    =&~  \lv \mathrm{Re}\l \G\G^*\left(\bm{L}_{\E} \bm{R}_{\E}^*-\bm{L} \bm{R}^*\right),\P_{\Omega_s}\left(\bm{L} \bm{R}^*-\bm{L}_{\E} \bm{R}_{\E}^*\right)\r\lv\\
    =&~  \lv \mathrm{Re}\l \left(\bm{L}_{\E} \bm{R}_{\E}^*-\bm{L} \bm{R}^*\right), \P_{\Omega_s}\G\G^*\P_{\Omega_s}\left(\bm{L} \bm{R}^*-\bm{L}_{\E} \bm{R}_{\E}^*\right)\r\lv\\
    = &~ \lV \G\G^*\P_{\Omega_s}\left(\bm{L}_{\E} \bm{R}_{\E}^*-\bm{L} \bm{R}^*\right)\rV_\fro^2 \le  \lV \P_{\Omega_s}\left(\bm{L}_{\E} \bm{R}_{\E}^*-\bm{L} \bm{R}^*\right)\rV_\fro^2,
\end{aligned}
\end{equation*}
where the second equality follows from $\Omega_s\subseteq \Omega$. The third equality follows from $\P_{\Omega_s}\G\G^*\P_{\Omega_s}X=\G\G^*\P_{\Omega_s}X$ for $X\in \mathbb{C}^{n_1\times n_2}$, which comes from the definition of $\G$ and $\G^*$ . The inequality in the last line follows from the fact $\G\G^*$ is a projection operator. By the definition of $\P_{\Omega_s}$ and $\Omega_s$, we know $\P_{\Omega_s}\left(\bm{M}\right)$ has at most $\gamma \alpha pn$ nonzero elements each row and each column. Then by \cref{PZUZV-M}, we have
\begin{equation}\label{T4}
T_4\le 18 c_s\gamma\alpha p \mu r\sigma_1^{\natural}\Delta.
\end{equation}
About $T_5$, by the definition of $\Gamma_{\gamma \alpha p}$, we know for $i\in \Omega\cap\left(\Omega_{\natural}\setminus \Omega_s\right)$, the term $z^{\natural}_i+s^{\natural}_i-[\G^*\left(\bm{L} \bm{R}^*\right)]_i$ is smaller than the $\gamma\alpha pn$-th largest element in $\Pi_{\Omega}\left(\bm{z}^{\natural}+\bm{s}^{\natural}-\G^*\left(\bm{L} \bm{R}^*\right)\right)$, then it is smaller than $\left(\gamma\alpha pn-\alpha pn\right)$-th largest element in $\Pi_{\Omega}\big(\bm{z}^{\natural}-\G^*(\bm{L} \bm{R}^*)\big)$ as $\lV \Pi_{\Omega}\bm{s}^{\natural}\rV_0\le \alpha pn$. It then implies for $i\in \Omega\cap\left(\Omega_{\natural}\setminus \Omega_s\right)$, we have
\begin{equation}\label{zzs}
\lv z^{\natural}_i+s^{\natural}_i-[\G^*\left(\bm{L} \bm{R}^*\right)]_i\lv^2 \le {\frac{\lV \Pi_{\Omega}\left(\bm{z}^{\natural}-\G^*\left(\bm{L} \bm{R}^*\right)\right)\rV_2^2}{\gamma\alpha pn-\alpha pn}}.
\end{equation}
Thus, for $i\in \Omega\cap\left(\Omega_{\natural}\setminus \Omega_s\right)$, it holds
\begin{equation*}
\begin{aligned}
\lv s_i^{\natural}[\overline{\bm{z}^{\natural}-\G^*\left(\bm{L} \bm{R}^*\right)}]_i\lv
=&~\lv [\bm{s}^{\natural}+\bm{z}^{\natural}-\G^*\left(\bm{L} \bm{R}^*\right)+\G^*\left(\bm{L} \bm{R}^*\right)-\bm{z}^{\natural}]_i [\overline{\bm{z}^{\natural}-\G^*\left(\bm{L} \bm{R}^*\right)}]_i\lv\\
\le &~ \lv [\bm{z}^{\natural}-\G^*\left(\bm{L} \bm{R}^*\right)]_i\lv^2+\lv [\bm{s}^{\natural}+\bm{z}^{\natural}-\G^*\left(\bm{L} \bm{R}^*\right)]_i\lv \lv [\overline{\bm{z}^{\natural}-\G^*\left(\bm{L} \bm{R}^*\right)}]_i\lv\\
\le &~(1+{\frac{\beta}{2}})\lv [\bm{z}^{\natural}-\G^*\left(\bm{L} \bm{R}^*\right)]_i\lv^2
+{\frac{1}{2\beta}}\lv z^{\natural}_i+s^{\natural}_i-[\G^*\left(\bm{L} \bm{R}^*\right)]_i\lv^2\\
\le &~ (1+{\frac{\beta}{2}})\lv [\bm{z}^{\natural}-\G^*\left(\bm{L} \bm{R}^*\right)]_i\lv^2+\frac{\lV \Pi_{\Omega}[\bm{z}^{\natural}-\G^*\left(\bm{L} \bm{R}^*\right)]\rV_2^2}{2\beta\left(\gamma-1\right)\alpha p n}
\end{aligned}
\end{equation*}
for any $\beta>0$, where the second inequality follows form $a^2+ab\le\left(\frac{\beta}{2}+1\right)a^2+\frac{b^2}{2\beta}$.
Denote $\Omega':=\Omega_{\natural}\setminus \Omega_s$. By the above inequality we obtain
\begin{align}\label{T5}
 &~T_5
 =\lv \mathrm{Re}\l\ \Pi_{\Omega'}\left(\bm{0}-\bm{s}^{\natural}\right),\Pi_{\Omega}[\G^*\left(\bm{L} \bm{R}^*\right)-\bm{z}^{\natural}]\r\lv\cr
    \le &~\left(1+{\frac{\beta}{2}}\right)\lV \P_{\Omega}\P_{\Omega'}\G\G^*\left(\bm{L} \bm{R}^*-\bm{L}_{\E} \bm{R}_{\E}^*\right)\rV_\fro^2
    +\sum_{i\in \Omega'}{\frac{\lV \P_{\Omega}\G\G^*\left(\bm{L} \bm{R}^*-\bm{L}_{\E} \bm{R}_{\E}^*\right)\rV_\fro^2}{2\beta\left(\gamma-1\right)\alpha p n}}\cr
    \le&~ \left(18+9\beta\right)c_s\alpha p\mu r\sigma_1^{\natural}\Delta+\frac{4\left(1+10\varepsilon_0\right)p\sigma_1^{\natural}\Delta+\frac{p}{2}\Delta^2}{2\beta\left(\gamma-1\right)},
\end{align}
where the first inequality follows from $\lV \G \bm{z}\rV_\fro^2=\lV\sum_{i=1}^n z_i/\sqrt{\varsigma_i}\H\bm{e}_i\rV_F^2=\lV \bm{z}\rV_2^2, \bm{z}\in\mathbb{C}^n$ and thus $\lV \G^*\bm{Z}\rV_2=\lV \G\G^*\bm{Z}\rV_\fro, \bm{Z}\in \mathbb{C}^{n_1\times n_2}$. The second inequality follows from  the same argument to \eqref{T4},  $\lvert\Omega'\rvert\le \alpha p n$ and \cref{PZ-Z}.

To obtain upper bound of $T_6$, the strategy is similar to $T_4$ and $T_5$. We have
\begin{align*}
T_6
 \le &~\lv\mathrm{Re}\l \P_{\Omega}\G \Pi_{\Omega_s}\left(\bm{s}-\bm{s}^{\natural}\right),\BDeltaL \BDeltaR^*\r\lv+ \lv\mathrm{Re}\l \P_{\Omega}\G \Pi_{\Omega_{\natural}\setminus \Omega_s}\left(\bm{s}-\bm{s}^{\natural}\right),\BDeltaL \BDeltaR^*\r\lv \cr
\le &~\lv\mathrm{Re}\l \G \Pi_{\Omega_s}\left(\G^*\left(\bm{L} \bm{R}^*\right)-\bm{z}^{\natural}\right),\P_{\Omega}\BDeltaL \BDeltaR^*\r\lv+ \lv\mathrm{Re}\l\G \Pi_{\Omega_{\natural}\setminus \Omega_s}\bm{s}^{\natural},\P_{\Omega}\BDeltaL \BDeltaR^*\r\lv \cr
    \le&~  \lV \P_{\Omega_s}\G\G^*\left(\bm{L}_{\E} \bm{R}_{\E}^*-\bm{L} \bm{R}^*\right)\rV_\fro \lV \P_{\Omega}\BDeltaL\BDeltaR^*\rV_\fro\cr
    &~+\lV\Pi_{\Omega_{\natural}\setminus \Omega_s}\left[\bm{z}^{\natural}+\bm{s}^{\natural}-\G^*\left(\bm{L} \bm{R}^*\right)\right]\rV_\fro  \lV\P_{\Omega}\BDeltaL \BDeltaR^*\rV_\fro\cr
   &~ +\lV\Pi_{\Omega_{\natural}\setminus \Omega_s}(\G^*(\bm{L} \bm{R}^*)-\bm{z}^{\natural})\rV_\fro\lV \P_{\Omega}\BDeltaL \BDeltaR^*\rV_\fro\cr
\le &~ \left((\sqrt{\gamma}+1)\sqrt{18c_s\alpha p\mu r\sigma_1^{\natural}\Delta}+{\frac{\lV \P_{\Omega}\G\G^*(\bm{L}           \bm{R}^*-\bm{L}_{\E} \bm{R}_{\E}^*)\rV_\fro}{\sqrt{\gamma-1}}}\right)\lV \P_{\Omega}\BDeltaL \BDeltaR^*\rV_\fro\cr
   \le &~ \left((\sqrt{\gamma}+1)\sqrt{18c_s\alpha p\mu r\sigma_1^{\natural}\Delta}+\textstyle{\sqrt{\frac{4(1+10\varepsilon_0)p\sigma_1^{\natural}\Delta+\frac{p}{2}\Delta^2}{\gamma-1}}}\right) \sqrt{{\frac{p}{4}}\Delta^2+18\varepsilon_0 p\sigma_1^{\natural}\Delta},
\end{align*}
where the fourth inequality follows from \eqref{T4}, \eqref{zzs} and the equality $\lV \G^*Z\rV_2=\lV \G\G^*Z\rV_\fro$. The fifth inequality follows from \cref{PZ-Z}. Let $\gamma$ be some constant satisfying $1+\frac{1}{b_0}\le\gamma\le 2$, where $1\le b_0<\infty$. Then, for all $\varepsilon_0\in\left(0,\frac{1}{10}\right)$, one has the upper bound of $T_6$ given by
\begin{equation*}
p^{-1}T_6\le \frac{\sqrt{2b_0}}{4}\Delta^2+\left(\sqrt{3b_0}+6\sqrt{c_s\alpha\mu r}\right)\sqrt{\sigma_1^{\natural}\Delta^3}
+\sqrt{\varepsilon_0}\left(12\sqrt{b_0}+44\sqrt{c_s\alpha \mu r}\right)\sigma_1^{\natural}\Delta,
\end{equation*}
and thus
\begin{align}
T_3\le&~ [\left(54 +9\beta\right)c_s\alpha\mu r+4b_0\beta^{-1}+\varepsilon_2]\sigma_1^{\natural}\Delta+ \left({\frac{\sqrt{2b_0}}{4}}+ \frac{b_0}{4\beta}\right)\Delta^2 \cr
&~+\left(\sqrt{3b_0}+6\sqrt{c_s\alpha\mu r}\right)\sqrt{\sigma_1^{\natural}\Delta^3},
\end{align}
where $\varepsilon_2=\sqrt{\varepsilon_0}\left(12\sqrt{b_0}+44\sqrt{c_s\alpha \mu r}\right)$, and thus $\varepsilon_1$ can be any small number in $\left(0,\theta\right)$ for some $\theta>0$ when $\alpha\mu r$ is bounded. Combining all the pieces, we obtain
\begin{align}\label{lowerboundF0}
&~T_1+T_2+T_3\ge T_2-\lv T_1\lv-\lv T_3\lv \cr
\ge&~\textstyle{\frac{15}{16}}\lV \bm{L} \bm{R}^*-\bm{L}_{\E} \bm{R}_{\E}^*\rV_\fro^2-\Big(2+ {\frac{\sqrt{2b_0}}{4}}+{\frac{b_0}{4\beta}}\Big)\Delta^2
 -\nu\sigma_1^{\natural}\Delta-\big(\sqrt{3b_0}+6\sqrt{c_s\alpha\mu r}\big)\sqrt{\sigma_1^{\natural}\Delta^3},
\end{align}
where we let $\nu=\left(54 +9\beta\right)c_s\alpha\mu r+4b_0\beta^{-1}+\varepsilon_1$ and $\varepsilon_1=\varepsilon_2+41\varepsilon_0$.
Now in the second part of the proof, we give a lower bound of $\l\nabla \phi,\bm{Z}-\bm{Z}^{\natural}\r$, where $\phi\left(\bm{L},\bm{R}\right)=\frac{1}{4}\lV \bm{L}^* \bm{L}-\bm{R}^* \bm{R}\rV_\fro^2$ is actually a standard regularization term. Therefore, by directly applying the result in \cite[Lemma 3]{yi2016fast}, we have
\begin{align}\label{lowerboundG}
&~\mathrm{Re}\left(\l \nabla_{\bm{L}} \phi,\BDeltaL\r+\l \nabla_{\bm{R}} \phi,\BDeltaR\r\right) \cr
\ge&~ {\frac{1}{4}}\lV\bm{L}^*\bm{L}-\bm{R}^*\bm{R}\rV_\fro^2+{\frac{1}{4}}\sigma_r^\natural\Delta-\sqrt{2\sigma_1^\natural\Delta^3}-\lV\bm{L}\bm{R}^*-\bm{L}_{\E}\bm{R}_{\E}^*\rV_\fro^2,
\end{align}
provided $(\bm{L},\bm{R}) \in \mathcal{B}(\sqrt{\sigma_r^\natural})$. Finally, taking $\lambda=\frac{1}{16}$, combining \eqref{lowerboundG} and the lower bound for $T_1+T_2+T_3$, we have
\begin{align*}
&~\mathrm{Re}\left(\l \nabla_{\bm{L}} \ell,\BDeltaL\r+\l  \nabla_{\bm{R}} \ell,\BDeltaR\r\right)\\
\ge&~\textstyle{\frac{7}{8}}\lV \bm{L} \bm{R}^*
-\bm{L}_{\E} \bm{R}_{\E}^*\rV_\fro^2+{\frac{1}{64}}\sigma_r^\natural\Delta-\nu\sigma_1^{\natural}\Delta-\left(2+ {\frac{\sqrt{2b_0}}{4}}+{\frac{b_0}{4\beta}}\right)\Delta^2\\
 &~-\left(\sqrt{3b_0}+6\sqrt{c_s\alpha\mu r}+\textstyle{\frac{\sqrt{2}}{16}}\right)\sqrt{\sigma_1^{\natural}\Delta^3}
 +\textstyle{\frac{1}{64}}\lV\bm{L}^*\bm{L}-\bm{R}^*\bm{R}\rV_\fro^2.
\end{align*}
\end{proof}

\begin{lemma}\label{T2up} Let $\BDeltaL,\BDeltaR, \bm{L}, \bm{R},\bm{L}_{\E},\bm{R}_{\E}$, and $\Delta$ be defined in \cref{def1}. Set $\lambda=\frac{1}{16}$. Let $\gamma\in \left[1+\frac{1}{b_0},2\right]$. Then under events \cref{event:RIP3}, \eqref{eq:projectionRIP1} and \eqref{eq:projectionRIP2}, for any $\varepsilon_0\in(0,\frac{1}{10})\cap \left(0,\frac{\sqrt{\kappa r}}{8c_0}\right)$ ,
if $p\ge \left(\varepsilon_0^{-2}c_4 c_s^2\mu^2 r^2 \log n\right)/n$, $\alpha\le \frac{1}{32 c_s\mu r\kappa}$, it holds
\begin{equation*}
\begin{split}
&~\lV \nabla_{\bm{L}} \ell\left(\bm{L},\bm{R};\bm{s}\right)\rV_\fro^2+ \lV \nabla_{\bm{R}} \ell\left(\bm{L},\bm{R};\bm{s}\right)\rV_\fro^2   \\
\le&~32\sigma_1^{\natural}\left((9b_0+10)\lV \bm{L} \bm{R}^*-\bm{L}_{\E} \bm{R}_{\E}^*\rV_\fro^2+72(b_0+1)\varepsilon_0 \sigma_1^{\natural}\Delta+{\frac{11}{2}}(b_0+1)\Delta^2\right)\\
\quad &~+{\frac{\sigma_1^{\natural}}{16}}\lV\bm{L}^*\bm{L}-\bm{R}^*\bm{R}\rV_\fro^2
\end{split}
\end{equation*}

\end{lemma}
\begin{proof}
To bound $\lV \nabla_{\bm{L}} \ell\rV_\fro^2+ \lV \nabla_{\bm{R}} \ell\rV_\fro^2=\lV \nabla_{\bm{L}} \psi+\lambda \nabla_{\bm{L}} \phi\rV_\fro^2+ \lV \nabla_{\bm{R}}\psi+\lambda \nabla_{\bm{R}}\phi\rV_\fro^2$, we first bound $\lV \nabla_{\bm{L}}\psi\rV_\fro^2+ \lV \nabla_{\bm{R}}\psi\rV_\fro^2$. And we have
\begin{align*}
&~\lV \nabla_{\bm{L}} \psi\rV_\fro^2 =\Big\lVert \Big(\frac{1}{p}\G\Pi_{\Omega}[\G^*\left(\bm{L} \bm{R}_{\E}^*\right)-\bm{z}^{\natural}+\bm{s}-\bm{s}^{\natural}]
 +\left(\I-\G\G^*\right)\left(\bm{L} \bm{R}^*-\bm{L}_{\E} \bm{R}_{\E}^*\right)\Big)\bm{R}\Big\lVert_\fro^2\\
\le&~ 2\lV\bm{R}\rV_2^2 \Big(\underbrace{ \lV p^{-1}\G\Pi_{\Omega}[\G^*\left(\bm{L} \bm{R}_{\E}^*\right)-\bm{z}^{\natural}+\bm{s}-\bm{s}^{\natural}]\rV_\fro^2}_{T_7}
+\underbrace{\lV \left(\I-\G\G^*\right)\left(\bm{L} \bm{R}^*-\bm{L}_{\E} \bm{R}_{\E}^*\right)\rV_\fro^2}_{T_8}\Big).
\end{align*}
Then we bound $T_7$ and $T_8$. Let $\Omega_s$ and $\Omega_s^{\natural}$ be the support of $\bm{s}$ and $\Pi_{\Omega}\bm{s}^{\natural}$, then we have
\begin{equation*}
\begin{split}
&\text{for} \ i\in \Omega\cap \Omega_s,~\bm{s}_i=[\bm{z}^{\natural}+\bm{s}^{\natural}-\G^*\left(\bm{L} \bm{R}_{\E}^*\right)]_i \Rightarrow [\G^*\left(\bm{L} \bm{R}_{\E}^*\right)-\bm{z}^{\natural}+\bm{s}-\bm{s}^{\natural}]_i=0,\\
&\text{for} \ i\in \Omega\cap\left(\Omega_s^{\natural}\setminus\Omega_s\right),~ [\bm{z}^{\natural}+\bm{s}^{\natural}-\G^*\left(\bm{L} \bm{R}_{\E}^*\right)]_i^2\le \textstyle{\frac{\lV \Pi_{\Omega}\left(\bm{z}^{\natural}-\G^*\left(\bm{L} \bm{R}^*\right)\right)\rV_2^2}{\gamma\alpha pn-\alpha pn}} , \\
&\text{for other}\ i \ \text{in}\ \Omega, ~\text{we have}\ [\bm{z}^{\natural}+\bm{s}^{\natural}-\bm{s}-\G^*\left(\bm{L} \bm{R}_{\E}^*\right)]_i=[\bm{z}^{\natural}-\G^*\left(\bm{L} \bm{R}_{\E}^*\right)]_i,
\end{split}
\end{equation*}
the equality in the first line follows from the definition of $\bs$ in \cref{def1}, the second line follows from \eqref{zzs}. Thus we have
\begin{align*}
T_7
\le&~ 2\lV p^{-1} \Pi_{\Omega\cap\left(\Omega_s^{\natural}\setminus\Omega_s\right)}[\G^*\left(\bm{L} \bm{R}_{\E}^*\right)-\bm{z}^{\natural}+\bm{s}-\bm{s}^{\natural}]\rV_2^2+2\lV p^{-1}\G\Pi_{\Omega}[\G^*\left(\bm{L} \bm{R}_{\E}^*\right)-\bm{z}^{\natural}]\rV_\fro^2\\
\le & ~2\alpha pn { \frac{\lV p^{-1} \Pi_{\Omega}\left(\bm{z}^{\natural}-\G^*\left(\bm{L} \bm{R}^*\right)\right)\rV_2^2}{\gamma\alpha pn-\alpha pn}}+2p^{-1}\lV \P_{\Omega}\G\G^*\left(\bm{L} \bm{R}^*-\bm{L}_{\E} \bm{R}_{\E}^*\right)\rV_\fro^2\\
\le &~p^{-1}{\frac{2\gamma}{\gamma-1}}\lV \P_{\Omega}\G\G^*\left(\bm{L} \bm{R}^*-\bm{L}_{\E} \bm{R}_{\E}^*\right)\rV_\fro^2\\
\le &~{\frac{4\gamma}{\gamma-1}}\left(1+\varepsilon_0\right)\lV \bm{L}_{\E} \BDeltaR^*+\BDeltaL \bm{R}_{\E}^*\rV_\fro^2+\frac{\gamma}{\gamma-1}\Delta^2+\frac{72\gamma}{\gamma-1}\varepsilon_0 \sigma_1^{\natural}\Delta\\
\le &~{\frac{8\gamma}{\gamma-1}}\left(1+\varepsilon_0\right)\lV \bm{L} \bm{R}^*-\bm{L}_{\E}\bm{R}_{\E}^*\rV_\fro^2+\frac{(5+4\varepsilon_0)\gamma}{\gamma-1}\Delta^2+\frac{72\gamma}{\gamma-1}\varepsilon_0 \sigma_1^{\natural}\Delta,
\end{align*}
where the first to third inequality follows from $\lV \G \bm{z}\rV_\fro^2=\lV \bm{z}\rV_2^2$, the fourth inequality follows from the third line in \eqref{ineq:projLR}, the last inequality follows from \eqref{Hu}. Since $\I-\G\G^*$ is a projection, we have
$$T_8\le \lV \bm{L} \bm{R}^*-\bm{L}_{\E} \bm{R}_{\E}^*\rV_\fro^2.$$
Thus, as $\gamma\in[1+\frac{1}{b_0},2]$ with $b_0\ge 1$, by setting $\varepsilon_0\in\left(0,\frac{1}{10}\right)$, combining $T_7$ and $T_8$ gives
\begin{equation*}
\begin{aligned}
\lV \nabla_{\bm{L}} \psi\rV_\fro^2
\le  2\lV\bm{R}\rV_2^2 \big((9b_0+10)\lV \bm{L} \bm{R}^*-\bm{L}_{\E} \bm{R}_{\E}^*\rV_\fro^2+72(b_0+1)\varepsilon_0 \sigma_1^{\natural}\Delta+{\frac{11}{2}}(b_0+1)\Delta^2\big).
\end{aligned}
\end{equation*}
One can follow the same process to bound $\lV  \nabla_{\bm{R}} \psi\rV_\fro^2 $, and thus we have
\begin{align}\label{upperboundF}
&~\lV \nabla_{\bm{L}} \psi\rV_\fro^2+ \lV \nabla_{\bm{R}} \psi\rV_\fro^2 \cr
\le&~  16\sigma_1^{\natural}\Big((9b_0+10)\lV \bm{L} \bm{R}^*-\bm{L}_{\E} \bm{R}_{\E}^*\rV_\fro^2+72(b_0+1)\varepsilon_0 \sigma_1^{\natural}\Delta+{\frac{11}{2}}(b_0+1)\Delta^2\Big),
\end{align}
where the inequality follows from $\lV\bm{L}\rV_2\le 2\sqrt{\sigma_1^{\natural}}$, $\lV\bm{R}\rV_2\le 2\sqrt{\sigma_1^{\natural}}$ as $(\bm{L}, \bm{R})\in \B(\sqrt{\sigma_1^{\natural}})$. Now we bound $\lV \nabla_{\bm{L}}\phi\rV_\fro^2+ \lV \nabla_{\bm{R}}\phi\rV_\fro^2$. In fact, we have
\begin{align}\label{upperboundG0}
&~\lV \nabla_{\bm{L}}\phi\rV_\fro^2+ \lV \nabla_{\bm{R}} \phi\rV_\fro^2=\lV\bm{L}\left(\bm{L}^*\bm{L}-\bm{R}^*\bm{R}\right)\rV_\fro^2+\lV\bm{R}\left(\bm{R}^*\bm{R}-\bm{L}^*\bm{L}\right)\rV_\fro^2 \cr
\le&~ \left(\lV\bm{L}\rV_2^2+\lV\bm{R}\rV_2^2\right)\lV\bm{L}^*\bm{L}-\bm{R}^*\bm{R}\rV_\fro^2\le  8\sigma_1^{\natural}\lV\bm{L}^*\bm{L}-\bm{R}^*\bm{R}\rV_\fro^2.
\end{align}
Finally, taking $\lambda=\frac{1}{16}$ and combining \eqref{upperboundF} and \eqref{upperboundG0} to obtain
\begin{align*}
&~\lV \nabla_{\bm{L}} \ell\rV_\fro^2+ \lV \nabla_{\bm{R}} \ell\rV_\fro^2
\le 2\left(\lV \nabla_{\bm{L}} \psi\rV_\fro^2+ \lV \nabla_{\bm{R}}\psi\rV_\fro^2\right)+2\lambda^2\left(\lV \nabla_{\bm{L}} \phi\rV_\fro^2+ \lV \nabla_{\bm{R}}\phi\rV_\fro^2\right)\\
\le&~32\sigma_1^{\natural}\left((9b_0+10)\lV \bm{L} \bm{R}^*-\bm{L}_{\E} \bm{R}_{\E}^*\rV_\fro^2+72(b_0+1)\varepsilon_0 \sigma_1^{\natural}\Delta+{\frac{11}{2}}(b_0+1)\Delta^2\right)\\
\quad &~+{\frac{\sigma_1^{\natural}}{16}}\lV\bm{L}^*\bm{L}-\bm{R}^*\bm{R}\rV_\fro^2
\end{align*}
\end{proof}
As discussed above, with \cref{T1low,T2up}, we are now ready to show the local descent property of the proposed method.

\begin{proof}[Proof of \cref{thm:convergence}]
The proof of the theorem is under events \eqref{eq:projectionRIP1} and \eqref{eq:projectionRIP2}. For ease of notation, let $\left(\bm{L}_{\E}^{(k)},\bm{R}_{\E}^{(k)}\right)\in \E\left(\bm{L}^{\natural},\bm{R}^{\natural}\right)$ defined to be matrices aligned with $\left(\bm{L}^{(k)},\bm{R}^{(k)}\right)$. Then, by the definition of $d_k$ we have
\begin{align}\label{error:dk}
d_{k+1}^2 =&~\lV \bm{L}^{(k+1)}-\bm{L}_{\E}^{(k+1)} \rV_\fro^2+\lV \bm{R}^{(k+1)}-\bm{R}_{\E}^{(k+1)} \rV_\fro \cr
    \le&~ \lV \bm{L}^{(k+1)}-\bm{L}_{\E}^{(k)} \rV_\fro^2+\lV \bm{R}^{(k+1)}-\bm{R}_{\E}^{(k)} \rV_\fro \cr
    \le&~ \lV \bm{L}^{(k)}-\eta\nabla_{\bm{L}} \ell^{(k)}-\bm{L}_{\E}^{(k)} \rV_\fro^2+\lV \bm{R}^{(k)}-\eta \nabla_{\bm{R}} \ell^{(k)}-\bm{R}_{\E}^{(k)} \rV_\fro\cr
    \le&~ d_k^2-2\eta\mathrm{Re}\left(\l \nabla_{\bm{L}} \ell^{(k)},\bm{L}^{(k)}-\bm{L}_{\E}^{(k)}\r+\l  \nabla_{\bm{R}} \ell^{(k)},\bm{R}^{(k)}-\bm{R}_{\E}^{(k)}\r\right)\cr
    &~+\eta^2\left(\lV\nabla_{\bm{L}} \ell^{(k)} \rV_\fro^2+\lV \nabla_{\bm{R}} \ell^{(k)} \rV_\fro^2\right),
\end{align}
where the second inequality comes from the non-expansion property of projection onto $\L$ and $\R$. 
Let $\Delta_k:=d_k^2$. Set $\eta\le\frac{\tilde{\theta}}{(b_0+1) \sigma_r^{\natural}}$ for sufficiently small constant $\tilde{\theta}$ and $\gamma_k\in \left[1 + \frac{1}{b_0},2\right]$ with $b_0\ge 1$. By \cref{T1low,T2up}, we then have
\begin{equation*}
\begin{split}
&~-2\eta\mathrm{Re}\left(\l \nabla_{\bm{L}} \ell^{(k)},\bm{L}^{(k)}-\bm{L}_{\E}^{(k)}\r+\l  \nabla_{\bm{R}} \ell^{(k)},\bm{R}^{(k)}-\bm{R}_{\E}^{(k)}\r\right)\cr
&~+\eta^2\left(\lV\nabla_{\bm{L}} \ell^{(k)} \rV_\fro^2+\lV \nabla_{\bm{R}} \ell^{(k)} \rV_\fro^2\right)\\
\le&~\eta\Big[(-\textstyle{\frac{7}{4}}+10\tilde{\theta})\lV \bm{L} \bm{R}^*
-\bm{L}_{\E} \bm{R}_{\E}^*\rV_\fro^2-(\frac{1}{32}\sigma_r^\natural-2\nu\sigma_1^{\natural}-72\tilde{\theta}\varepsilon_0\sigma_1^\natural)\Delta_k\\
 &~+\left(4+ \textstyle{\frac{\sqrt{2b_0}}{2}}+\frac{b_0}{2\beta}+\frac{11\tilde{\theta}}{2}\right)\Delta_k^2+\left(\sqrt{12b_0}+12\sqrt{c_s\alpha\mu r}+\textstyle{\frac{\sqrt{2}}{8}}\right)\sqrt{\sigma_1^{\natural}\Delta_k^3}
 \\
 &~-(\textstyle{\frac{1}{32}}-\frac{\tilde{\theta}}{64})\lV\bm{L}^*\bm{L}-\bm{R}^*\bm{R}\rV_\fro^2\Big]\\
\le &~ -{\frac{1}{64}}\eta\sigma_r^{\natural}\Delta_k,
\end{split}
\end{equation*}
where the in second inequality, we use $\Delta_k \le \frac{\theta_1}{b_0}\frac{\sigma_r^{\natural}}{\kappa}$ wtih $ \frac{\theta_1}{b_0}=c_1^2$ sufficiently small.  Also, in the inequality we set $\beta=\frac{\kappa}{\theta_2}$, $\varepsilon_0=\frac{\theta_3}{\kappa}$, $\varepsilon_1=\frac{\theta_4}{\kappa}$, and assume $\alpha\le\frac{\theta_5}{c_s\mu r \kappa^2}$ for sufficiently small constants $\theta_2,\theta_3,\theta_4,\theta_5$. Then, it implies
\begin{equation*}
d_{k+1}^2\le \left(1-{\frac{\eta \sigma_r^{\natural}}{64}}\right)d_{k}^2.
\end{equation*}
Note that under events \cref{event:RIP3}, \eqref{eq:projectionRIP1} and \eqref{eq:projectionRIP2}, the above inequality holds for all $k\ge 0$. 
\end{proof}

\section{Supporting lemmas}\label{supportlemmas}
\begin{lemma}[{\cite[Lemma~6]{cai2021asap}}] \label{s-matrix}
For any $\bm{z}\in\mathbb{C}^{n}$ such that $\|\bm{z}\|_0\leq\alpha n$, it holds
\begin{equation*}
\lV \H \bm{z}\rV_2 \leq \alpha n \lV \H\bm{z}\rV_{\infty} =\alpha n\lV \bm{z}\rV_{\infty}.
\end{equation*}
\end{lemma}
\begin{lemma}\label{uv}
For all $\bm{u}\in \mathbb{R}^{n_1}$, $\bm{v}\in \mathbb{R}^{n_2}$, it holds
\begin{equation*}
\frac{1}{p}\sum_a\sum_{i+j=a+1}\delta_a\bm{u}_i\bm{v}_j\le \lV\bm{u}\rV_1 \lV\bm{v}\rV_1+\sqrt{\frac{8n\log n}{p}}\lV\bm{u}\rV_2 \lV\bm{v}\rV_2
\end{equation*}
with probability at least $1-2n^{-2}$, provided $p\ge \left( \log n\right)/n$.
\end{lemma}
\begin{proof}
The result of this lemma is similar to \cite[Lemma 5]{cai2018spectral} except the slightly different sampling models. We provide the proof for completeness. Let $\widehat{\bm{H}}_a:=\sqrt{\varsigma_a}\bm{H}_a=\H \bm{e}_a$, then
\begin{equation*}
\begin{aligned}
\frac{1}{p}\sum_a\sum_{i+j=a+1}\delta_a\bm{u}_i\bm{v}_j=& \frac{1}{p}\sum_a \delta_a\bm{u}\bm{H}_a\bm{v}\\
=&~\frac{1}{p}\bm{u}^\top\left(\sum_a \left(\delta_a-p\right)\bm{H}_a\right)\bm{v}+\bm{u}^\top \left(\bm{1}_{n_1} \bm{1}^\top_{n_2}\right)\bm{v}\\
\le&~\lV\left(\sum_a \left(\frac{\delta_a}{p}-1\right)\bm{H}_a\right)\rV_2 \lV\bm{v}\rV_2 \lV\bm{v}\rV_2+\lV\bm{u}\rV_1 \lV\bm{v}\rV_1.
\end{aligned}
\end{equation*}
Denote $\bm{\R}_a:=\left(\frac{\delta_a}{p}-1\right)\bm{H}_a$. We have $\mathbb{E}[\bm{\R}_a]=0$, and
\begin{equation*}
\lV \bm{\R}_a\rV_2 \le \frac{1}{p}\lV \bm{H}_a\rV_2\le \frac{1}{p}.
\end{equation*}
Moreover, we have
\begin{equation*}
\lV \mathbb{E}\left[\sum_a\bm{\R}_a \bm{\R}_a^*\right]\rV_2=\lV \sum_a \mathbb{E}\left(\frac{\delta_a}{p}-1\right)^2\bm{H}_a\bm{H}_a^* \rV_2\le \frac{n}{p}\lV\bm{H}_a\bm{H}_a^* \rV_2\le\frac{n}{p}.
\end{equation*}
Similarly, $\lV \mathbb{E}\left[\sum_a\bm{\R}_a^*\bm{\R}_a \right]\rV_2\le\frac{n}{p}$.
Hence, by Bernstein's inequality \cite[Theorem 1.6]{tropp2012user}, we have
\begin{equation*}
\mathbb{P}\left(\lV \sum_a \bm{\R}_a\rV_2>t\right)\le \left(n_1+n_2\right)\exp\left(\frac{-p t^2/2}{n+t/3}\right).
\end{equation*}
Letting $t=\sqrt{\left(8 n\log n\right)/p}$, we have
\begin{equation*}
\mathbb{P}\left(\lV \sum_a \bm{\R}_a\rV_2>t\right)\le 2n^{-2},
\end{equation*}
provided $p\ge \left( \log n\right)/n$.
\end{proof}

\section*{Acknowledgment}
The work of HanQin Cai is partially supported by NSF DMS 2304489. The work of J.-F. Cai is partially supported by Hong Kong Research Grants Council (HKRGC) GRF grants 16309518, 16309219, 16310620, and 16306821.
\bibliographystyle{abbrv}
\bibliography{robustHMR}





\end{document}